\documentclass[11pt]{article}
\usepackage[margin=1in]{geometry}

\usepackage{wrapfig}
\usepackage{caption,subcaption}
\usepackage{amsmath}
\usepackage{amsthm}
\usepackage{thm-restate}
\usepackage{graphicx}
\usepackage[utf8]{inputenc} % allow utf-8 input
\usepackage[T1]{fontenc}    % use 8-bit T1 fonts
\usepackage{hyperref}       % hyperlinks
\usepackage{url}            % simple URL typesetting
\usepackage{booktabs}       % professional-quality tables
\usepackage{amsfonts}       % blackboard math symbols
\usepackage{nicefrac}       % compact symbols for 1/2, etc.
\usepackage{microtype}      % microtypography
\usepackage{xcolor}         % colors

\usepackage{color}

\newtheorem{theorem}{Theorem}
\newtheorem{lemma}[theorem]{Lemma}
\newtheorem{corollary}[theorem]{Corollary}

\newtheorem{obs}[theorem]{Observation}

\newcommand{\E}{{\mathbf E}}
\newcommand{\opt}{{\tt OPT}}
\newcommand{\UCB}{{\tt UCB}}
\newcommand{\all}{{\sc AllProbe}\ }
\newcommand{\mmax}{{\sc BestProbe}\ }

\newcommand{\D}{\mathcal{D}}
\newcommand{\hwc}{{\sc HwC}\ }
\newcommand{\btrl}{{\sc BtRL}\ }
\newcommand{\W}{\mathcal{W}}

\newcommand{\lap}{{\sc LwC}\ }

\title{Online Learning and Bandits with Queried Hints} 
\author{
  Aditya Bhaskara\thanks{Supported by NSF grant CCF-2047288.} \\
  School of Computing \\
  University of Utah \\
  \texttt{bhaskaraaditya@gmail.com} \\
  \and
  Sreenivas Gollapudi \\
  Google Research \\
  Mountain View, CA \\
  \texttt{sgollapu@google.com} \\
  \and
  Sungjin Im\thanks{Supported in part by NSF grants  CCF-1844939 and CCF-2121745.} \\
  University of California \\
  Merced, CA \\
  \texttt{sim3@ucmerced.edu}
  \and
  Kostas Kollias \\
  Google Research \\
  Mountain View, CA \\
  \texttt{kostaskollias@google.com}
  \and
  Kamesh Munagala\thanks{Supported by NSF grant CCF-2113798.} \\
  Computer Science Department \\
  Duke University \\
  \texttt{kamesh@cs.duke.edu} \\
}
\date{}

\begin{document}

\maketitle

\begin{abstract}
We consider the classic online learning and stochastic multi-armed bandit (MAB) problems, when at each step, the online policy can \emph{probe} and find out which of a small number ($k$) of choices has better reward (or loss) before making its choice. In this model, we derive algorithms whose regret bounds have  exponentially better dependence on the time horizon compared to the classic regret bounds. In particular, we show that probing with $k=2$ suffices to achieve time-independent regret bounds for online linear and convex optimization. The same number of probes improve the regret bound of stochastic MAB with independent arms from $O(\sqrt{nT})$ to $O(n^2 \log T)$, where $n$ is the number of arms and $T$ is the horizon length. For stochastic MAB, we also consider a stronger model where a probe reveals the reward values of the probed arms, and show that in this case, $k=3$ probes suffice to achieve parameter-independent constant regret, $O(n^2)$. Such regret bounds cannot be achieved even with full feedback {\em after} the play, showcasing the power of limited ``advice'' via probing {\em before} making the play. We also present extensions to the  setting where the hints can be imperfect, and to the case of stochastic MAB where the rewards of the arms can be correlated.
\end{abstract}

\section{Introduction}
In this paper, we consider two problems that form the cornerstone of sequential analysis and decision theory,  a field first developed by Wald~\cite{Wald} in the 1940's. The first is the online linear/convex optimization problem that was initially studied in the context of repeated games by Hannan~\cite{Hannan} and Blackwell~\cite{Blackwell} in the 1950's. In this problem, there is a possibly infinite space of potential actions in a high-dimensional space. At each step, a decision maker needs to choose one action, that is sometimes called an ``arm''. Subsequently, nature presents an adversarially chosen linear (or convex) loss function, and the decision maker incurs the evaluation of this loss function at the chosen action. Subsequently, the decision maker is told the loss function at that time step. The goal is to compete with an omniscient policy that knows all the loss functions in advance, but is restricted to choosing one fixed action for all time steps. The difference between the loss of the policy that that of the omniscient policy is termed {\em regret}. Over the decades, several policies have been developed~\cite{KalaiV,LittlestoneW,Zinkevich} that achieve regret $O(\sqrt{T})$, where $T$ is the horizon length, and the $O(\cdot)$ hides problem-dependent parameters. Such a dependence on $T$ is also optimal~\cite{hazan}. A well-known specialization is the {\em experts problem} where the actions or arms form a discrete set, and each arm incurs an arbitrary loss at each time step that is unrelated to the losses of the other arms~\cite{LittlestoneW,howto}. 

The second problem is the stochastic multi-armed bandit (MAB) problem that was first formulated by Robbins~\cite{Robbins} around 1950, though the widely used Thompson Sampling policy for this problem dates back to the 1930's~\cite{Thompson}. In this problem, a decision maker is faced with $n$ independent arms that yield {\em i.i.d.} rewards from unknown underlying distributions. As before, the goal is to design a policy or allocation rule to sequentially play these arms to maximize reward. In a sense, this can be viewed as a stochastic version of the experts problem; however, the key difference is that the decision maker only learns the reward of the chosen arm at the end of the time step, and not the rewards of all the arms. The regret is measured against an omniscient policy that knew the reward distributions of the arms (but not the reward values), and chooses the single arm with maximum expected reward at all time steps.  In seminal work, Lai and Robbins~\cite{LaiR} showed an optimal allocation policy along with tight lower bounds on the regret incurred by any policy, assuming a parametric form on the distributions. This result was subsequently generalized by Auer, Cesa-Bianchi, and Fisher~\cite{AuerCF}, who derived similar upper bounds without assuming a parametric form on the distributions. 

Both the online convex optimization and the stochastic MAB problem have found numerous applications in areas ranging from clinical trial design to ad-word allocations to recommendation systems, and continue to be extensively studied in the fields of statistics and machine learning. For more on the history and variants of this problem, we refer the reader to several excellent books~\cite{BubeckBook,SlivkinsBook,LattimoreBook,hazan}.

\subsection{Probe Model and Motivation} 
In this paper, we study the following twist on these problems. Suppose before playing, we are allowed to query or {\em probe} an oracle with $k$ options, and the oracle responds with the best of these at that step -- either telling the algorithm the identity of this option and nothing more, or telling the algorithm the rewards/losses of all the options. We subsequently play the option suggested by the oracle. How should these probes be chosen, and can we obtain much better regret? 

Our motivation for such a probe model comes from the recent literature on designing algorithms that can leverage machine learning (ML) based {\em predictions}. This paradigm has been used to obtain improved guarantees for many classic online algorithmic problems~\cite{RakhlinSridharan,lykouris2018competitive,LattanziLMV20,purohit2018improving,GollapudiP19} (see also the survery~\cite{MitzenmacherV20}). In these settings, the online algorithm is assumed to have access to an auxiliary ML model that predicts properties of the arriving inputs. The goal is to derive improved bounds assuming the predictions are correct, whilst doing nearly as well as worst-case algorithms when the predictions are incorrect. 
Of particular relevance is recent work on ``parsimonious hints'', where the ML model provides as little information to the online algorithm as possible. For instance, recent work has considered online linear optimization, where the hint is a direction with a strictly positive dot product with the cost vector~\cite{Dekel,bhaskara21}. In such a model, for optimization over a sphere, the regret bound improves from $\Omega(\sqrt{T})$ to $O(\ln T)$, even with hints at  $O(\sqrt{T})$ time steps. Similarly, recent work~\cite{im22a} has considered randomized caching with parsimonious hints about next request time of a few cached pages. 

The probe model can now be viewed as a parsimonious ML hint, where the algorithm queries the predictor with a few options, asking it either for the identity of the best option among these, or for the losses/rewards of all these options.  As an example, consider modeling the problem of shortest path routing using stochastic bandits (where each of $m$ paths has its length drawn from a distribution at each step). In this case, a routing engine can query for the length of a few paths, by possibly querying users who have chosen this path, before routing other users. 

\subsection{Results and Conceptual Contribution}
In this paper, we study the following question:
\begin{quote}
Can we obtain improved regret bounds for online learning and MAB problems against the classic best arm in hindsight benchmark, if the algorithm is allowed the power of using ML advice via a few probes before making its decision to play?
\end{quote}

We answer this question in the affirmative: We show that just $k = 2$ probes (or $k=3, 4$ for stronger results)  suffice to drastically improve the known regret guarantees. Indeed, our main results (in Sections~\ref{sec:experts} and~\ref{sec:all}) show a {\em constant regret bound independent of the time horizon}, assuming the hints are correct! 

Before proceeding further, it is instructive to compare our results with~\cite{Dekel}. They consider online linear optimization where the domain is the sphere (or more generally, strictly convex) and show that if the hint is a direction with positive ``correlation'' with the cost vector, the regret improves to $O(\ln T)$. They also show that such an improvement cannot be obtained if the decision space is more general, say the $\ell_{\infty}$ ball, where the regret remains $\Omega(\sqrt{T})$. In contrast, we show that if the algorithm is allowed to choose a direction and ask the predictor if the cost along that direction is increasing or not, then the regret improves to constant; further this result holds for linear optimization over {\em any domain} (not just strictly convex), and convex optimization over {\em any} convex space.

In the literature on algorithms with ML-based predictions, one key requirement is for the algorithms to not be ``thrown off'' by incorrect predictions. This aspect is referred to as robustness~\cite{MitzenmacherV20}, and the goal is to recover worst-case guarantees (ones possible without any hints) even if the hints are adversarial. In the context of linear optimization, this was studied in~\cite{bhaskara20a}. They extend the result of~\cite{Dekel} again assuming the optimization is over a sphere; however, their result also has a dependence of $O(\ln T)$ on the horizon length. As our second contribution, with $B$ imperfect hints, our probe model improves this to obtain regret $O(\sqrt{B+1})$ (with no dependence on the horizon $T$), again holding for any underlying space over which the linear optimization is performed, as well as for convex optimization over any convex space. We finally note that there was no extension known for~\cite{Dekel} to the MAB problem, and one of our contributions is to develop stochastic bandit algorithms with improved regret bounds under parsimonious hints.  

At a conceptual level, our work shows that 
querying or probing options via comparisons is far more powerful than more `passive' models for hints. We believe that such query based hint models may find other applications beyond online learning. At an even higher level, our work is reminiscent of the \emph{power of two choices} in online load balancing~\cite{AzarBKU}, where choosing the lesser loaded of two random bins leads to an exponential improvement in the expected maximum load.  This paradigm has found numerous applications, such as hashing, congestion control, and distributed memory management. Our paper is in a similar vein – we show that allowing a few queries or hints suffices to give an exponential improvement in regret. 

\subsection{Overview and Technical Highlight} 
We present our model and summary of results in Section~\ref{sec:model}, where we also place the regret bounds we obtain in context.  Our main probing model is the \mmax model, where the online policy probes a set of arms or options ahead of the play, and is told the best of these options (without revealing the actual losses/rewards of these options). For the MAB problem, we also consider the \all model where the policy also observes the rewards of all arms it probes. We present algorithms and regret bounds for online linear/convex optimization in the \mmax model in Section~\ref{sec:experts}. Our results extend to the case where the hints can be imperfect, and we present this in Section~\ref{sec:imperfect}. We present analogous results for the MAB problem in Section~\ref{sec:max}. We consider the MAB problem in the \all model in Sections~\ref{sec:all} and~\ref{sec:mab-dependent}.

At a technical level, our results require development of new probabilistic tools to lower bound (resp. upper bound) the maximum (resp. minimum) of independent random variables 
(See Lemmas~\ref{lem:min},~\ref{lem:min_gen},~\ref{lem:max} and~\ref{lem:pairs}.) We term these as ``reverse prophet inequalities'', since they are in some sense the reverse of well-known prophet inequality results~\cite{Krengel,SamuelCahn} that upper bound the expected maximum of a set of independent  random variables by a sum of quantities related to individual distributions. These technical lemmas are crucial to both our algorithms as well as our improved regret analyses. Much like prophet inequalities, these lemmas are of independent interest as stand-alone probability tools.

\subsection{Other Related Work} 

\medskip \noindent {\bf Stochastic Probing.} The question of adaptively or non-adaptively probing independent distributions has been widely studied, with applications to  database query optimization~\cite{GoelGM06,MunagalaBMW05,DeshpandeHK16,Liu}, wireless communication~\cite{GuhaMS06}, and traffic routing~\cite{BhaskaraGKM20}. Much like our model, a probe reveals the true underlying value drawn from the distribution; however, this line of work largely focuses on algorithm design as opposed to learning. It is shown in~\cite{GoelGM06} that the problem of computing the best set of $k$ distributions to probe in order to maximize the expected value of the maximum of the probed set is {\sc NP-Hard} when $k$ is not a constant.  A related problem is the {\em Pandora's problem}~\cite{Pandora,GuhaMS06,BeyhaghiK19}, where there is no bound on the number of probes, but we seek to maximize the largest value found minus the total probing cost spent in discovering the value. %
A general adaptive greedy algorithm for such problems, which probes the next distribution conditioned on the values seen so far, was presented in~\cite{GolovinK}. Our work is different in that we assume $k$ is a small constant (so that {\sc NP-Hardness} is not an issue); instead, we seek to understand the power of such probes in repeated bandit interactions.

\medskip \noindent {\bf Bandits with Probes.} In the {\em cascading bandits} model~\cite{Cascade1,Slivkins}, a recommendation system (such as search engine or streaming service) needs to choose $k$ items (or arms) to show to a user and obtains feedback on what item the user clicks on. Similarly, in the bandits with pre-observation model~\cite{CJW20}, the arms are wireless channels of unknown quality, and a user needs to sequentially probe the channels until a good channel is found. This is a bandit version of the Pandora's problem. Finally, motivated by job scheduling applications, the online budgeted submodular coverage problem~\cite{StreeterG} considers the more general problem where the reward obtained by the player is an unknown submodular function of the probed arms. Such problems have also been recently considered in the experts setting~\cite{garbage} where the policy gets feedback about all arms after the play, but is allowed to make a bounded number of probes before the play.  One commonality in all the above works is that they consider {\em policy regret}: The benchmark is an omniscient solution that not only knows the rewards/losses of the arms, but is also allowed as many probes per step as the online policy. 
In contrast, motivated by machine learning hints, we consider probing as providing the online policy {\em more power} compared to an omniscient benchmark and our goal is to study the resulting improvement in regret bounds. 

\medskip \noindent {\bf Dueling Bandits.} Our problem is also related to {\em dueling bandits}~\cite{Yue12}, where the only information available on playing a pair of arms is the result of a noisy comparison. Assuming the noisy comparison model is Condorcet consistent (e.g., the Bradley-Terry model), the goal is to minimize the error incurred in not playing the best arm. Though our problem is superficially similar in that we allow plays of multiple arms, in our case, we observe the reward of at least one of the played arms. Further, it is easily possible to incur zero (or even negative) regret in our problem by playing two sub-optimal arms. This makes our problem technically very different. Nevertheless, it is an interesting question whether techniques from dueling bandits can be used to improve some of our results.

\medskip \noindent {\bf Predictable Sequences.} In the field of online algorithms in general, recent research has focused on incorporating machine learning predictions to  obtain more optimistic bounds if the predictions are correct, but  preserve the robustness of the classic online model in case the predictions  turn out to be inaccurate. Of particular relevance is recent work~\cite{RakhlinSridharan,WeiLuo,SteinhardtLiang} on bandits and experts with ``predictable sequences'', where improved regret bounds are shown when the algorithm is given a prediction that is close on average to the true reward at each step. Our work is similar in spirit if we view the probe as a perfect prediction; however, in our case, the probes are both {\em interactive} and {\em parsimonious}, meaning that the policy itself has to decide on the predictions to obtain each step, and further, these are few in number.

\medskip \noindent {\bf Bandits with Limited Advice.} A related model to predictable sequences is the {\em experts} model, where in each step, each of $m$ experts makes a prediction about the best arm at that step. In the limited advice setting~\cite{Seldin,Kale}, the policy can choose only $k \ll m$ of the experts in each round and obtain their predicted best arm; subsequently, the policy plays one of the arms and obtains its reward. The goal of the policy is to compete with reward of following the best expert in hindsight. However, the predictions of these experts can be arbitrary, and unrelated to which arm was actually the best arm. Therefore, even if a policy obtains the predictions of all experts, it cannot avoid $\Omega(\sqrt{nT})$ regret, making this problem fundamentally different from our setting.

\section{Model and Results}
\label{sec:model}
In the online linear optimization problem, we are given a a finite set $\W \subseteq \mathbb{R}^d$ of options (or arms).  Assume that for all $w \in \W$, we have $w \in [-1,1]^d$. At each step $t$, the algorithm chooses action $w^t$ and is subsequently presented a cost vector $\ell^t \in [-1,1]^d$, for which it incurs loss $\langle \ell^t, w^t \rangle \in [-d,d]$. 
Note that the choice of action $w^t$ only depends on the cost vectors $\ell^q$ for $q < t$.  We assume cost vectors $\ell^t$ are generated by an {\em oblivious} adversary that does not know the internal randomness used by the algorithm in choosing $w^t$. We further assume that any linear function can be efficiently minimized over the set $\W$. 
We describe the related experts and online convex optimization settings in Appendix~\ref{app:extension1}.

There is a horizon of $T$ steps; we assume $T$ is known, but this assumption can be removed using standard techniques. We measure regret with respect to the  hindsight optimum as:
$$ \mbox{Regret} = \sum_{t=1}^T \langle \ell^t, w^t \rangle - \min_{w \in \W} \sum_{t=1}^T \langle \ell^t, w \rangle.$$ 
We denote $\opt = \min_{w \in \W} \sum_{t=1}^T \langle \ell^t, w \rangle$. 

In the stochastic MAB problem, there are $n$ bandit arms. Arm $i$ yields $i.i.d.$ rewards $X_i$ drawn from an independent distribution $D_i \in [0,1]$. Let $\mu_i = \E[X_i]$. The policy can play one arm $i_t$ at each step $t$, this choice can depend on the observed rewards till time $t-1$. Unlike online linear optimization, in the MAB problem, at the end of time step $t$, the policy {\em only} learns the reward $r_{{i_t}t} \sim D_{i_t}$ of the arm $i_t$ that it plays at time $t$. The hindsight optimum plays the arm with highest expected reward each step so that $\opt = T \max_i \E[X_i]$. We consider pseudo-regret, which is simply $\opt - \E[\sum_{t=1}^T r_{{i_t} t}]$. In Section~\ref{sec:mab-dependent}, we extend this to the case where there is a joint (correlated) reward distribution over arms and the rewards are drawn $i.i.d.$ across time from this joint distribution.

\medskip \noindent {\bf Probe Model.} There are two models of probing and feedback that we consider. These models are parameterized by a number $k \le n$, which captures the number of probes allowed. The first model applies to both online learning and stochastic MAB, while the second applies only to MAB.

\begin{description}
\item[Probes with Best Arm Feedback (\mmax).] Any policy probes a set $S_t$ of at most $k$ arms or options at any time step $t$, and learns which arm in $S_t$ will incur the lowest loss or maximum reward at step $t$ (but not the reward/loss values). The policy plays the arm $i \in S_t$ with largest reward (resp. minimum loss). For instance, in the MAB problem, if the reward of arm $i$ at time $t$ at step $t$ is $r_{it} \sim D_{i}$, the policy's reward is $\hat{R}_t = \max_{i \in S_t} r_{it}$.  
\item[Probes with All Feedback (\all).] For the stochastic MAB problem, we also consider a model that gives the online policy more information. In this model, the policy is actually told the rewards $r_{it} \sim D_i$ for each arm $i \in S_t$, and it subsequently plays the arm with largest reward, incurring reward $\hat{R}_t = \max_{i \in S_t} r_{it}$ at that step.  \end{description}

When $k = 1$, these models reduce to the classic versions of online linear optimization and stochastic MAB. Note that our regret measure is a departure from work on policy regret: In our case, $\opt$ remains the same whether $k=1$ or $k > 1$. In that sense, we seek to understand the power of increasing $k$ on the regret of the online policy. As mentioned before, the motivation comes from viewing the $k> 1$ setting as ML hints received by the online policy.

\subsection{Our Results}
Our main results are as follows:

\begin{itemize}
    \item In Section~\ref{sec:experts}, for the online linear optimization problem in the \mmax model, we show a regret bound of $O\left(d^2 \ln d\right)$, independent of $T$ with $k=2$ probes. Using similar techniques, in Appendix~\ref{app:extension1}, we improve to $O(\ln n)$ for the experts model with $n$ experts, and we also present a horizon-independent regret bound for online convex optimization.
    \item In Section~\ref{sec:imperfect}, we extend the model to the setting where $B$ of the hints could be imperfect. We show an algorithm with regret bound $O\left(d^2 \ln d \sqrt{B+1}\right)$. The dependence on $B$ is trivially optimal, as can be seen by considering $B=T$. 
    Similar results can be obtained for online convex optimization.
    \item In Section~\ref{sec:max}, for the stochastic MAB problem in the \mmax model, using $k=2$ probes, we obtain a {\em parameter independent} regret bound of $O(n^2 \log T)$. 
    \item In Section~\ref{sec:all}, for the stochastic MAB problem in the \all model, we show a regret bound of $O(n^2)$, independent of $T$, with only $k = 3$ probes. 
    \item In Section~\ref{sec:mab-dependent}, we consider the stochastic MAB problem in \all model when the rewards of arms can be correlated at each time step (whilst still being independent across time). Here, we show a $\tilde{O}(n^{8/3} T^{1/3})$ regret bound using $k=4$ probes.
\end{itemize}

Our regret bounds in Sections~\ref{sec:max} and~\ref{sec:all} for the MAB problem are {\em parameter independent}, meaning that there is no dependence of the regret bound on the means and variances of the individual arms. All of our bounds except those in Section~\ref{sec:mab-dependent} are an exponential improvement in $T$ over bounds without probes, since there is a lower bound of $\Omega(\sqrt{nT})$~\cite{AuerCFS} for the stochastic MAB problem, and a lower bound of $\Omega(\sqrt{T \ln n})$~\cite{howto} for the experts problem with $k=1$ probes.  

We also remark that the difference between the last two results is a subtle one. Typically in the stochastic MAB problem, the rewards of the arms are assumed to be independent, but most  known results carry over to the case where the rewards can be dependent (as long as the samples are independent across time). Informally, this is because the algorithm only receives a feedback about one arm at each step. This is not the case when $k>1$, and leads to the main open problem of understanding if we can improve upon the $\Theta(\sqrt{T})$ regret for adversarial bandits in the \all{} or \mmax{} models.

\section{Online Linear Optimization in the \mmax Model}
\label{sec:experts}
We now consider online linear optimization in the \mmax model. Recall that in standard online linear optimization, the algorithm needs to play an option (or arm) $w \in \W$ each step, and learns the linear loss function $\ell^t$ at the end of each step $t$. We show that with an oracle that can return the better of $k=2$ options at each step $t$, we can achieve regret that is independent of $T$. 

For simplicity, we will assume $\W$ is a convex polytope in $\mathbb{R}^d$, though our result easily extends to the case where $\W$ is any finite set of options over which linear functions can be efficiently optimized. In more generality, our results hold as long as there are a finite set of options $S \subseteq \W$, such that for $w \in S$, the set $C_w = \{ \vec{v} \in \mathbb{R}^d | w = \mbox{argmin}_{s \in \W} \langle s, \vec{v} \rangle \}$ has positive volume, and further, $\cup_{w \in S} C_w = \mathbb{R}^d$.  

Our results also extend to the more general case of online convex optimization with arbitrary convex domains. For conceptual  simplicity, we present only the linear case here and defer the general case to  Appendix~\ref{app:extension-oco}. In the {\em special case} of online learning with experts (where the domain $\W$ is the $n$-dimensional unit simplex), we present an improved regret bound that only depends logarithmically on the number of experts, whilst still being independent of $T$. The logarithmic dependence on the number of experts is analogous to the standard regret bounds for the experts problem, and the details are presented in Appendix~\ref{app:extension-experts}. 

Finally, all of these algorithms can be extended to the case when hints/probes can be incorrect at a small number of steps. Once again, we describe the algorithm only for online linear optimization, in Section~\ref{sec:imperfect}.

\subsection{Algorithm: Differentially Private Regularization} 
Let us first recap the algorithmic framework of randomized regularization~\cite{KalaiV} for the setting without probes. Let $L^{t-1} = \sum_{q=1}^{t-1} \ell^q$ denote the sum of the cost vectors till time $t$. The algorithm chooses a $d$-dimensional random cost vector $x$ of sufficiently large variance upfront and at step $t$, chooses the {\em regularized optimum} action,
$$ w^t = \mbox{argmin}_{w \in \W} \langle L^{t-1} + x, w \rangle.$$
The analysis proceeds in two parts. First it is shown that if $L^t$ were hypothetically used instead of $L^{t-1}$ in the above step, the only regret would be due to adding noise $x$, and this is independent of the time horizon $T$. Next, it is shown that since $x$ has large variance, using $L^{t-1}$ instead of $L^t$ produces almost the same distribution of the regularized optimum $w^t$. These steps trade-off, since the larger the variance of noise, the worse the first step and better the second. The optimal trade-off yields a $O(\sqrt{T})$ bound. 

In the probing model, our algorithm \lap{} will simply sample {\em two} random vectors $x,y$ and compute the regularized optimal solutions as above. The algorithm will find out which of these solutions has smaller loss at time $t$, and then choose this solution as its action $w^t$.  Our key lemma (Lemma~\ref{lem:min}) shows that if the noise vector is chosen so as to satisfy a {\em differential privacy} property, then the error in the first step above (comparing with $L^t$) goes away! 
In other words, the better of two samples produced using the regularized distribution obtained using $L^{t-1}$ will be as good as a sample obtained using  $L^t$.  

We note that the use of differentially private noise was first considered in~\cite{Abernathy1}, who observed that viewing randomized regularization as differential privacy of the loss across time leads to simpler analysis and somewhat stronger regret bounds. We show an {\em algorithmic} application of this approach for probing. Indeed, though classical regularization does not require differentially private noise (indeed, it does not even require randomization~\cite{Zinkevich}), this seems critical to achieving our bounds. 

\medskip
\noindent {\bf Algorithm.} We now formally describe the algorithm. Recall that the distribution {\tt Laplace}$(\beta)$ has density function $f(x) = \frac{1}{2\beta} \exp(-|x|/\beta)$ for $x \in \mathbb{R}$.  For $\eta \le 0.4$ being a constant, the algorithm performs these steps at time $t$. 

\begin{itemize}
\item Choose $x_j \sim {\tt Laplace}(d/\eta)$ for each $j \in \{1,\ldots,d\}$; set $a^t = \mbox{argmin}_{w \in \W} \langle L^{t-1} + x, w \rangle$. 
\item Choose $y_j \sim {\tt Laplace}(d/\eta)$ for each $j \in \{1,2,\ldots,d\}$; set $b^t = \mbox{argmin}_{w \in \W} \langle L^{t-1} + y, w \rangle$. 
\item Let $A^t = \langle a^t, \ell^t \rangle $ and $B^t = \langle b^t, \ell^t \rangle $. Probe to learn $w^t = \mbox{argmin}_{a^t, b^t} \{ A^t, B^t \}$. 
\item Play $w^t$  as the action at time $t$, incurring actual loss $\min(A^t, B^t)$. 
\end{itemize}

We will show the following theorem:
\begin{theorem}
\label{thm:main1}
The constant $\eta \in (0,0.4]$, the regret of the \lap algorithm is $O\left(d^2 \ln d\right)$.
\end{theorem}

\subsection{Analysis} 
As in the classic analysis of regularization~\cite{KalaiV}, define a hypothetical ``Be the Regularized Leader'' (\btrl) algorithm: Choose $x_j \sim {\tt Laplace}(d/\eta)$ independently for each $j \in \{1,2,\ldots,d\}$. At step $t$, use $c^t = \mbox{argmin}_{w \in \W} \langle L^{t-1} + \ell^t + x, w \rangle$ as the action taken at step $t$. Note that \btrl is {\em not realizable}. %
Let $D = \max_{w,w' \in \W} | w - w'|_1$. The next lemma restates the classic ``be the leader'' result from~\cite{KalaiV}. 
\begin{lemma}[\cite{KalaiV}]
\label{lem:kv}
For any $\eta \ge 0$, the regret of the \btrl algorithm is at most $D \cdot \E[\max_{j=1}^d |x_j|]$.
\end{lemma}

For the specific setting of $x_j \sim  {\tt Laplace}(d/\eta)$, we have $\E[\max_{j=1}^d |x_j|] = O\left(\frac{d}{\eta} \ln d \right)$ assuming $\eta$ is a constant. Further, we have $D = 2d$, so that we obtain:

\begin{corollary}
\label{cor10}
For constant $\eta > 0$, the regret of the \btrl algorithm is $O\left(\frac{d^2 \ln d}{\eta}\right)$
\end{corollary}

In the rest of the analysis, we focus on a particular step $t$, and omit the superscript $t$. Let $\D_1$ denote the distribution of the regularized optimum $a^t$ (resp. $b^t$) using $L^{t-1}$ in the \lap algorithm, and $\D_2$ denote the distribution of the regularized optimum $c^t$ using $L^t$ in the \btrl algorithm. The following lemma is a consequence of the well-known Laplace mechanism in differential privacy~\cite{RothBook}, and we present a proof for completeness.

\begin{lemma}[$\eta$-Differential Privacy]
\label{lem:dp}
For all $w \in \W$, we have: $ \exp(- \eta) \le \frac{\Pr[\D_1 = w]}{\Pr[\D_2 = w]} \le \exp(\eta). $
\end{lemma}
\begin{proof}
Consider $x,y \sim {\tt Laplace}(d/\eta)$. Let $X_i$ and $Y_i$ denote the random variables $L^{t-1}_i + x_i$ and $L^{t-1}_i + \ell^t_i + y_i$ respectively. For any fixed value $v$, for any dimension $i$, their density functions are related as:
$$ \frac{f_{X_i} (v)}{f_{Y_i} (v)} \le \exp(\ell^t_i \eta/d) \le \exp(\eta/d),$$
since $\ell^t_i \le 1$. Let $X$ denote the $d$-dimensional random variable whose $i^{th}$ dimension is $X_i$, and similarly define $Y$.  Using the above, if $\vec{x}, \vec{y}$ have components drawn independently from ${\tt Laplace}(d/\eta)$, then for any $\vec{v} \in \mathbb{R}^d$, the density functions of $X$ and $Y$ are related as
$$ \frac{f_X(\vec{v})}{f_Y(\vec{v})} \le \prod_{i=1}^d  \exp(\ell^t_i \eta/d)  \le \exp(\eta).$$
A similar argument shows $ \frac{f_Y(\vec{v})}{f_X(\vec{v})} \le \exp(\eta)$. Since $\W$ is a convex polytope in $\mathbb{R}^d$, the optimum solution for any $\vec{v}$ is achieved at a vertex $w \in \W$. Further, w.l.o.g., the set of $\vec{v}$ whose optimum corresponds to $w \in \W$ define a convex cone in $\mathbb{R}^d$ with positive volume.  The lemma follows by integrating the density functions $f_X, f_Y$ over this cone.\footnote{Note that the proof also extends to the case where $\W$ is any finite set of outcomes, since w.l.o.g., there is a subset $S \subseteq \W$ of outcomes such that the set of $\vec{v}$ for which $w \in S$ is the optimal outcome is continuous and convex (and hence has positive volume), and the union over $S$ of these sets of $\vec{v}$ spans $\mathbb{R}^d$.}
\end{proof}

We will now overload notation and use $\D_1$ (resp. $\D_2$) to refer to the distribution of the losses $\langle \ell^t, w \rangle$ for $w$ chosen according to $\D_1$ (resp. $\D_2$).  Note that these new distributions are discrete with support size equal to the number of $w \in \W$ that are optimum for some $\vec{v} \in \mathbb{R}^d$, and also satisfy the previous lemma.  

The crux of the analysis is the following lemma, which shows that the expected min of the two losses of the regularized optima using $L^{t-1}$ is at most that of the regularized optimum using $L^t$, that is, the per-step loss of \lap is at most that of the \btrl algorithm. We can view this as a ``reverse prophet inequality'' that upper bounds the expected minimum instead of lower bounding it.

\begin{lemma}[Reverse Prophet Inequality for Private Noise]
\label{lem:min}
Let $A,B$ be losses drawn independently from $\D_1$ and let $C$ be a loss drawn from $\D_2$, where $\D_1, \D_2$ are of bounded support and satisfy $\eta$-differential privacy (Lemma~\ref{lem:dp}) for $\eta \in (0,0.4]$. Then $\E[\min(A,B)] \le \E[C]$.
\end{lemma}
\begin{proof}
By shifting the distributions if necessary, we may assume that the support of $\D_1$ and $\D_2$ is the set of non-negative real numbers. Let $G(x) = \Pr[\D_1 \ge x]$ and $\hat{G}(x) = \Pr[\D_2 \ge x]$. Similarly, let $F(x) = \Pr[\D_1 \le x]$ and $\hat{F}(x) = \Pr[\D_2 \le x]$. Let $q = \mbox{argmin}_x \{x | G(x) \le \exp(-\eta)\}$. Therefore,
$$ \E[C] = \int_{x=0}^q (1- \hat{F}(x)) dx + \int_{x=q}^{\infty} \hat{G}(x) dx = q - \int_{x=0}^q \hat{F}(x) dx + \int_{x=q}^{\infty} \hat{G}(x) dx,$$
$$ \E[\min(A,B)] = q - \int_{x=0}^q F(x) (2-F(x)) dx + \int_{x=q}^{\infty} (G(x))^2 dx.$$
Note that for $x \ge q$, $G(x) \le \exp(-\eta)$. Since $\hat{G}(x) \ge \exp(-\eta) G(x)$, this implies $(G(x))^2 \le \hat{G}(x)$. Further, for $x \in [0,q]$, we have $F(x) \le 1 - \exp(-\eta)$, and further, $\hat{F}(x) \le \exp(\eta) F(x)$. Therefore
$$ F(x) (2- F(x)) \ge \exp(-\eta) (1+\exp(-\eta)) \hat{F}(x) \ge \hat{F}(x)$$
for $\eta \in (0,0.4]$. Putting this together, we infer $\E[C] \ge \E[\min(A,B)]$, completing the proof.
\end{proof}

We note that the proof of the above lemma crucially needs the two-sided bound in Lemma~\ref{lem:dp}. In contrast, the classical regret results for random regularization, for instance, in~\cite{KalaiV}, only require the total variation distance between $\D_1$ and $\D_2$ be at most $\eta$. However, this weaker condition is insufficient for proving the lemma, as can be seen by the following example: $\D_1$ is a deterministic value $d$, while $\D_2$ is $d$ with probability $1-\eta$ and $0$ otherwise. Then, $\E[\min(A,B)] = \E[A] = d$, while $\E[C] = d (1-\eta)$. This leads to a regret of $\eta d T$ over $T$ steps against the \btrl algorithm.

\paragraph{Proof of Theorem~\ref{thm:main1}.} The proof is now immediate. Lemma~\ref{lem:min} implies that at step $t$, the expected loss of \lap is at most that of the \btrl algorithm. Using linearity of expectation over time steps and combining with the regret bound for \btrl from Corollary~\ref{cor10}, we have proved Theorem~\ref{thm:main1}.

\subsection{Handling Imperfect Hints}
\label{sec:imperfect}
We now consider the setting where at most $B$ of the $T$ hints (comparisons) yield incorrect answers. We show an algorithm that yields regret $O(d^2  \ln d \sqrt{B+1})$. At one extreme, when $B = T$, this recreates the $O(\sqrt{T})$ regret guarantee for classical online linear optimization, while at the other extreme, when $B = 0$, this recovers Theorem~\ref{thm:main1}. It is also easy to show that such a dependence on $B$ is optimal for any $T$. To see this, simply construct an instance where the loss function at steps where the hints are correct is identically zero, so that the hints are vacuous. The only relevant steps are the ones where the hints are incorrect, so that any algorithm's regret is lower-bounded by the regret of classical online linear optimization over $B$ steps.

\medskip\noindent{\bf Algorithm.} In the sequel, we assume $B$ is known; the case for unknown $B$ follows by the standard doubling trick where we maintain a guess for $B$ and restart the algorithm once $B$ doubles. This can be done since we are in the full information regime and we get to know if a query answer was incorrect; we omit the details. Our algorithm is nearly identical to \lap for $B = 0$, and proceeds as follows. We will set $\eta = \frac{1}{5 \sqrt{B+1}}$ in this algorithm. We also have a parameter $p = 5 \eta$. The main difference in the algorithm is the following: After \lap probes to learn $w^t = \mbox{argmin}_{a^t, b^t} \{ A^t, B^t \}$, the new algorithm plays $w^t$  as the action with probability $p$ and plays $a^t$ with probability $1-p$. In other words, the algorithm now uses the hint (action $w^t$) with probability $1-p$, else it ignores the hint at that step and mimics classical follow the regularized leader (action $a^t$). Note that at the end of every step, the algorithm learns if the hint was correct. This allows the algorithm to keep track of the number of mis-predictions and is critical to using the doubling trick for unknown $B$.

\medskip\noindent{\bf Analysis.} We will show the following theorem. We note that the generalization to the experts and online convex optimization settings (as described in Appendix~\ref{app:extension1}), to yield regret with a $O(\sqrt{B+1})$ dependence on $B$, follows the same outline and is hence omitted.

\begin{theorem}
\label{thm:main_gen}
The regret of the modified \lap algorithm with $B$ incorrect hints is $O(d^2  \ln d \sqrt{B+1})$.
\end{theorem}

Our analysis hinges on the following generalization of the reverse prophet inequality (Lemma~\ref{lem:min}) to mimic the behavior of the algorithm.

\begin{lemma}[(General Reverse Prophet Inequality.)]
\label{lem:min_gen}
Let $p>0$ be a given parameter, and let $\D_1, \D_2$ be bounded support probability distributions (over losses) that satisfy $\eta$-differential privacy (Lemma~\ref{lem:dp}) for some $\eta < p/4$. Let $C$ be a random sample from $\D_2$, and let $Z$ be a random variable obtained as follows: two samples $A, B$ are drawn from $\D_1$. With probability $(1-p)$, $Z$ is set to be $A$. With probability $p$, $Z = \min(A, B)$. Then we have  $\E[Z] \le \E[C]$.
\end{lemma}
\begin{proof}
Define $G(x)$ and $\hat{G}(x)$ as in Lemma~\ref{lem:min}. Once again, suppose that the support of the distributions is the set of non-negative reals. By definition, we have
\[ \E[C] = \int_{0}^\infty \hat{G}(x) dx \quad \text{and} \quad \E[Z] = \int_{0}^\infty \left( (1-p) G(x) + p G(x)^2 \right)dx.\]
By the differential privacy property and the choice of $\eta$, we have that $G(x), \hat{G}(x)$ are within a factor of $(1 \pm \frac{p}{2})$ of one another for every $x$. Now consider two cases. 

First, suppose $G(x) \le 1/2$. In this case, $(1-p) G(x) + p G(x)^2 \le (1-\frac{p}{2}) G(x) \le \hat{G}(x)$. 

Next, suppose $G(x) > 1/2$. Now, writing $F(x) = (1-G(x))$ for convenience, we have 
\begin{align*} (1-p) G(x) + p G(x)^2 &= (1- F(x))(1-p F(x)) \\ 
&= 1 - F(x) - pF(x) + pF(x)^2 \\ 
& < 1 - F(x) - \frac{p}{2} F(x) = 1 - (1+\frac{p}{2}) F(x).
\end{align*}
In the last inequality, we used $F(x) < 1/2$. By the privacy property, the final expression is $\le 1- \hat{F}(x) = \hat{G}(x)$.
Plugging this into the above integral completes the proof of the lemma.
\end{proof}

We now prove Theorem~\ref{thm:main_gen} by bounding the regret against the \btrl algorithm.

\begin{proof}[Proof of Theorem~\ref{thm:main_gen}]
Consider the \btrl algorithm with the same value of $\eta = \frac{1}{5 \sqrt{B+1}}$. By Corollary~\ref{cor1}, this has regret $O(d^2  \ln d \sqrt{B+1})$. We now bound the regret of our algorithm against the \btrl algorithm. Towards this end, set $S_1$ denote the set of time steps where the hints are correct. By Lemma~\ref{lem:min_gen}, the loss of our algorithm at these steps is at most that of the \btrl algorithm. 

Let $S_2$ denote the set of steps where the hints are incorrect. At each of these steps, with probability $p$ the algorithm plays $w^t$, and incurs loss $O(d^2)$. There are $B \cdot p = O(\sqrt{B+1})$ such steps in expectation, yielding total loss (and regret) at most $O(d^2 \sqrt{B+1})$ against the \btrl algorithm. For the remaining steps, the algorithm plays action $a^t$. Since Laplace noise is $\eta$-differentially private (Lemma~\ref{lem:dp}), the probability that $a^t \neq c^t$ is at most $\eta$, where $c^t$ is the action taken by the \btrl algorithm. In the case where $a^t \neq c^t$, the algorithm incurs regret $O(d^2)$. Therefore the total regret against \btrl due to these steps is $O(B \eta d^2) = O(d^2 \sqrt{B+1})$. Combining with the regret of $O(d^2  \ln d \sqrt{B+1})$ of the \btrl algorithm, this shows our algorithm also has regret $O(d^2  \ln d \sqrt{B+1})$ against the best fixed action in hindsight. This completes the proof.
\end{proof}

\section{Stochastic MAB in the \mmax Model with $k=2$ Probes}
\label{sec:max}
Recall that in the stochastic MAB problem in the \mmax model, the policy can probe $k$ arms and learn the identity of the arm with the maximum reward. At the end of the play, it only learns the reward of the arm that was played that step, and not the rewards of the other arms that were probed that step. We will show a parameter independent regret bound of $O(n^2 \log T)$ for a horizon of $T$ steps using only $k = 2$ probes. As mentioned before, this is an exponential improvement over the $\Omega(\sqrt{nT})$ regret necessary with $k=1$ probes. 

\subsection{The {\tt Meta UCB-V} Algorithm} 
\label{sec4-alg}
The key algorithmic idea is the following: If we use an optimistic estimate of the sample mean (similar to UCB~\cite{Robbins}), probe the top two arms based on this estimate and choose the arm with the higher reward, we obtain a guaranteed advantage (in expectation) over simply choosing the top arm based on the estimate. Our key technical lemmas, Lemma~\ref{lem:max} and~\ref{lem:pairs} show that if the gap between the means of the arms played and the mean of the optimal arm is small, the expected regret is actually $\le 0$. This is key to achieving our improved regret bounds.

Formally, define a new classical bandit instance (with no probes) as follows. For every pair of arms $(i,j)$, we have a meta-arm $(i,j)$. Therefore, the new instance as ${n \choose 2}$ meta-arms. If $X_i$ and $X_j$ denote the random variables corresponding to the rewards of arms $i$ and $j$ respectively, the reward of meta-arm $(i,j)$ is $\max(X_i, X_j)$. Playing the meta-arm $(i,j)$ corresponds to probing the pair of arms $i$ and $j$ and obtaining/observing the value $X_{ij} = \max(X_i, X_j)$. 

Our algorithm runs the UCB-V policy~\cite{Audibert} on the new bandit instance. We call this algorithm {\tt Meta UCB-V}. %
For completeness, this algorithm works as follows: At step $t$, suppose $(i,j)$ has been played $s^t_{ij}$ times. Let $m^t_{ij}$ and $V^t_{ij}$ denote the sample mean and sample variance over the $s^t_{ij}$ plays:
\begin{equation}
    \label{eq:sample}
    m_{ijt} = \frac{\sum_{q=1}^{s_{ijt}} X_{ijq}}{s_{ijt}} \qquad V_{ijt} = \frac{\sum_{q=1}^{s_{ijt}} (X_{ijq} - m_{ijt})^2}{s_{ijt}}.
\end{equation}
Define the quantity $\UCB^t_{ij}$ as: $ \UCB^t_{ij} = m^t_{ij} + \sqrt{ \frac{2.4 V^t_{ij} \log t}{s^t_{ij}}} + \frac{3.6 \log t}{s^t_{ij}}.$
At time step $t$, the meta-arm $(i,j)$ with the highest value of $\UCB^t_{ij}$ is played. 

 We will show the following parameter-independent regret bound as our main result; recall that we are considering pseudo-regret throughout this paper.

\begin{theorem}
\label{thm2}
 The {\tt Meta UCB-V} algorithm has regret $O(n^2 \log T)$  with $k=2$ probes.
\end{theorem}

Recall that $\mu_i = \E[X_i]$, $\mu^* = \max_i \E[X_i]$, and the benchmark is $\opt = T \mu^*$. We now define analogous quantities for the meta-arms. Recall that $X_{ij} = \max(X_i,X_j)$ is the random variable corresponding to the reward of meta-arm $(i,j)$. Let $\mu_{ij} = \E[X_{ij}]$ and $\sigma^2_{ij} = \mbox{Var}[X_{ij}]$. Let $M^* = \max_{(i,j)} \mu_{ij}$ and $\Delta_{ij} = M^* - \mu_{ij}$. 

For the {\tt Meta UCB-V} policy, let $T_{ij}$ denote the expected number of times meta-arm $(i,j)$ is played. Let $R_{ij} = T_{ij} \Delta_{ij}$ denote the expected regret against $M^*$ due to playing meta-arm $(i,j)$. The main result of Audibert {\em et al.}~\cite{Audibert} is the following.\footnote{We note that the result in~\cite{Audibert} assumes the arms are independent, while our meta-arms are correlated. However, the expected regret bounds in~\cite{Audibert} hold as is when the arms are correlated, provided the samples from the arms are {\em i.i.d.} across time. This suffices for our purposes.}
\begin{equation}
\label{eq:regret2}
    R_{ij} \le 10 \left( \frac{\sigma^2_{ij}}{\Delta_{ij}} + 1 \right) \ln T.
\end{equation}

\subsection{Reverse Prophet Inequalities}
Our main technical lemmas {\em lower bound} the expected maximum of a pair of random variables in terms of the mean and variance of the individual variables. These lemmas effectively show that if the gap between the means of the arms played and the mean of the optimal arm is small, the expected regret is actually $\le 0$. As mentioned before, these can be viewed as the reverse of standard prophet inequalities~\cite{Krengel,SamuelCahn} that upper bound the expected maximum.  These lemmas forms the crux of our analysis both in this section and in Section~\ref{sec:all}, and may be of independent interest in related settings.

\begin{lemma}
\label{lem:max}
Let $X$ and $Y$ be independent random variables supported on $[0,1]$ whose means satisfy $\mu_Y \ge \mu_X$, and let $Z = \max(X, Y)$. Then $\E[Z] \ge \mu_X + \sigma^2_X/2$.
\end{lemma}
\begin{proof}
We can write $\E[Z] = \int_{t=0}^1 \E[Z|X=t] f_X (t) dt, $
where $f_X$ is the pdf of $X$. We split the integral into two sums, $t \in [0,\mu_X]$ and $t \in [\mu_X, 1]$. For the first integral, we use $Z \ge Y$, and thus
$$ \int_{t=0}^{\mu_X} \E[Z | X=t] f_X (t) dt \ge \int_{t=0}^{\mu_X} \E[Y] f_X (t) dt \ge \int_{t=0}^{\mu_X} \mu_X f_X(t) dt. $$
Here, the first inequality uses the independence of $X$ and $Y$ and the second inequality uses $\mu_Y \ge \mu_X$. For the second integral, we use $Z \ge X$, and obtain
$$\int_{t=\mu_X}^1 \E[Z|X=t] f_X (t) dt \ge  \int_{t=\mu_X}^1 t f_X (t) dt. $$

Thus, writing the $t$ as $\mu_X +( t-\mu_X)$, we get that the sum of the two integrals is bounded as
\begin{equation}
\label{eq:bound1}
\E[Z] \ge \mu_X + \E[ (X - \mu_X)_+], 
\end{equation}
where $(X - \mu_X)_+$ is the random variable $\max(0, X-\mu_X)$. 

Likewise, define $(X-\mu_X)_- = \max(0, \mu_X - X)$. By the definition of the expectation, we have $\E[(X-\mu_X)_+] = \E[(X-\mu_X)_-]$.  On the other hand, because our random variables are bounded on $[0,1]$, we have $\E[(X-\mu_X)_+] \ge \E[(X-\mu_X)_+^2]$ and likewise for $\E[(X-\mu_X)_-]$. Therefore,
\begin{equation}
\label{eq:bound2}
    \E[(X-\mu_X)_+] + \E[(X-\mu_X)_-] \ge \E[ (X-\mu_X)^2] = \sigma^2_X.
\end{equation}
Thus both the terms are at least $\sigma^2_X/2$, implying the lemma.
\end{proof}

We next extend Lemma~\ref{lem:max} to pairs of arms. Note that this does not follow from Lemma~\ref{lem:max} by simply replacing ``arms" with ``pairs of arms", since different pairs of arms are no longer independent. 

\begin{lemma}
\label{lem:pairs}
Let pair $(i,j)$ be such that $\mu_{ij} < \mu^*$. Then $M^* - \mu_{ij} \ge \frac{\sigma^2_{ij}}{4}.$
\end{lemma}
\begin{proof}
For notational convenience, let $q = \mbox{argmax}_{s} \mu_s$. By assumption, we have $\E[X_q] = \mu^* > \mu_{ij}.$

Since $M^* = \max_{(q,r)} \mu_{qr}$, we have $M^* \ge \max(\mu_{iq}, \mu_{jq})$. Therefore,
\begin{equation}
    \label{eq1}
    M^* - \mu_{ij} \ge  \max(\mu_{iq}, \mu_{jq}) - \mu_{ij}.
\end{equation} 
We now use the same argument as the proof of Eq~(\ref{eq:bound1}). Set $Z = \max(X_i,X_q)$ and $X = X_i$,  but split the integral at $t = \mu^*$ instead of $t = \E[X_i]$. This yields
$$ \mu_{iq} = \E[Z] \ge \int_{t=0}^{\mu^*} \mu^* f_{X_i} (t) dt + \int_{t=\mu^*}^1 ((t - \mu^*) + \mu^*) f_{X_i} (t) dt  = \mu^* + \E[(X_i - \mu^*)_+], $$
A similar inequality holds for $\mu_{jq}$. Combining with Eq~(\ref{eq1}), we have
$$ M^* - \mu_{ij} \ge \max \left\{\E[(X_i - \mu^*)_+], \E[(X_j - \mu^*)_+] \right\} + (\mu^* - \mu_{ij}).$$
Since $\mu^* \ge \mu_{ij}$ by assumption, the above implies
\begin{align*} 
M^* - \mu_{ij} \ge & \max \left\{\E[(X_i - \mu_{ij})_+], \E[(X_j - \mu_{ij})_+] \right\} \\
 \ge & \frac{1}{2} \left(\E[(X_i - \mu_{ij})_+] + \E[(X_j - \mu_{ij})_+] \right) \ge  \frac{1}{2} \E[(X_{ij} - \mu_{ij})_+]
\end{align*}
To see the last inequality, simply observe that for any value $a \ge 0$, we have $\Pr[X_{ij} \ge \mu_{ij} + a] \le \Pr[X_i \ge \mu_{ij} + a] + \Pr[X_j \ge \mu_{ij} + a]$. 
Finally, using Eq~(\ref{eq:bound2}) with $X = X_{ij}$, we obtain
$$ M^* - \mu_{ij} \ge \frac{1}{2} \E[(X_{ij} - \mu_{ij})_+] \ge \frac{1}{4} \sigma^2_{ij}.$$
This completes the proof.
\end{proof}

\subsection{Proof of Theorem~\ref{thm2}}
We split the set of meta-arms into two types. The first type is arms $(i,j)$ for which $\mu_{ij} \ge \mu^*$. Since $\mu^* - \mu_{ij} \le 0$, the expected regret (against the benchmark $\mu^*$) of playing $(i,j)$ is non-positive.

We therefore focus on arms $(i,j)$ with $\mu_{ij} < \mu^*$. By Lemma~\ref{lem:pairs}, we have $\Delta_{ij} =  M^* - \mu_{ij} \ge \frac{\sigma^2_{ij}}{4}$. Combining with Eq~(\ref{eq:regret2}), we have $ R_{ij} \le 10(4+1) \ln T = 50 \ln T$. However, we have $\Delta_{ij} \ge \mu^* - \mu_{ij}$, so that the regret of $(i,j)$ against $\mu^*$ is at most that against $M^*$. Putting all this together, we have:
\begin{align*}
\mbox{Regret}  = & \sum_{(i,j) :\mu_{ij} < \mu^*}  T_{ij} (\mu^* - \mu_{ij})  \le  \sum_{(i,j) :\mu_{ij} < \mu^*}  T_{ij} (M^* - \mu_{ij}) \\
     = & \sum_{(i,j) :\mu_{ij} < \mu^*}  T_{ij} \Delta_{ij} =  \sum_{(i,j) :\mu_{ij} < \mu^*}  R_{ij} \le  \sum_{(i,j) :\mu_{ij} < \mu^*}  50 \ln T  =  O(n^2 \ln T)
  \end{align*}
This completes the proof of Theorem~\ref{thm2}.

\section{Stochastic MAB in the \all Model}
\label{sec:all}
We will now switch to the \all model for the stochastic MAB problem, where a policy can probe $k$ arms, and receive as feedback  the rewards of all these arms at the current time step.  In this model, we will  show a regret bound of $O(n^2)$ when we are allowed $k=3$ probes per step. 

In addition to using Lemma~\ref{lem:max} from the previous section, the main observation is that with $k=3$ probes allowed, we can have one probe dedicated to {\em exploring} the arms in a round-robin fashion. Thus \all{} allows us to overcome the explore/exploit trade-off and obtain better concentration bounds, albeit with the loss of an additional $n$ factor in the regret. 
The overall algorithm is reminiscent of the \UCB-V algorithm of Audibert {\em et al.}~\cite{Audibert} from Section~\ref{sec:max},  but differs in  how we construct the optimistic estimate, as well as the analysis.

\medskip
\noindent {\bf Simultaneous Explore-Exploit Algorithm.}
Our algorithm uses one probe per step to play the $n$ arms in a round robin fashion. We call this the {\em exploration probe} and it enables the algorithm to maintain the sample mean and sample variance for each arm. Therefore, at time $t \ge 1$, any arm $i$ is observed $s_{it} \ge \lfloor \frac{t}{n} \rfloor$ steps by the exploration probe. Let $X_{i1}, X_{i2}, \ldots, X_{is_{it}}$ denote these observations. Analogous to Eq.~(\ref{eq:sample}), let $m_{it}, V_{it}$ denote the sample mean and variance of arm $i$ after $s_{it}$ observations.
Choose $\epsilon = 0.1$; any other small constant will work equally well. Let 
\begin{equation}
\label{eq:UCB1}
    \UCB_{it} = m_{it} + \epsilon V_{it}.
\end{equation} 
The remaining two probes at time step $t$ are used to probe the two arms with the largest  values of $\UCB_{it}$. We call these the {\em exploitation} probes. This completes the description of the algorithm.

Note that in our analysis, we will assume the policy finally plays one of the two exploitation-probed arms and not the exploration-probed arm. We use the exploration probe only to update the estimates of $m_{it}$ and $V_{it}$, and the results of these probes could be obtained at the end of that time step. This is sufficient to get constant regret. 
We show the following theorem below.
\begin{theorem} 
\label{thm1}
The regret of the simultaneous explore-exploit algorithm is $O(n^2)$ for the \all model with $k=3$ probes.
\end{theorem}

We remark that the bound has a slightly worse dependence on $n$, the number of arms, than the standard UCB bounds~\cite{BubeckBook}. Improving the bound to $O(n)$ is an interesting open direction. Further, it is easy to show examples where playing the top two arms with highest UCB1 scores~\cite{AuerCF}, or running Thompson Sampling twice~\cite{agrawalG}, has regret that is polynomial in $T$. However, we conjecture that the policy that simply plays the top two arms according to $m_{it}$ also has constant regret, and we leave showing this as an interesting open question.

\subsection{Some Tail Bounds}
Before presenting the proof of Theorem~\ref{thm1}, we present some well-known tail bounds for the sample mean and variance. 

We will use tail bounds that follow from  Audibert {\em et al.}~\cite{Audibert,Freedman}. Fix a time $t$ and some arm $i$ whose distribution $D_i$ has mean $\mu_i$ and variance $\sigma^2_i$. Note that we assume $D_i \in [0,1]$. Consider the $s_{it}$ exploration probes and the resulting estimates $m_{it}$ and $V_{it}$ of the sample mean and variance respectively:
\begin{equation}
    \label{eq:sample2}
    m_{it} = \frac{\sum_{q=1}^{s_{it}} X_{iq}}{s_{it}} \qquad V_{it} = \frac{\sum_{q=1}^{s_{it}} (X_{iq} - m_{it})^2}{s_{it}}.
\end{equation}

\begin{theorem}[Implicit in Audibert {\em et al.}\cite{Audibert}]
\label{thm:tail}
For any time $t \geq 1$, we have the following tail bounds on $m_{it}$ and $V_{it}$ after $s_{it}$ probes of arm $i$.
\begin{enumerate}
    \item For $q \ge \frac{1}{18}$, we have: $ \Pr\left[ |m_{it} - \mu_i| > q \sigma^2_i \right] \le 3 e^{- \frac{q \sigma^2_i s_{it}}{23}}.$
    \item For $q \ge \frac{1}{18}$, we have: $ \Pr\left[ V_{it} > (1+q) \sigma^2_i \right] \le 3 e^{- \frac{q \sigma^2_i s_{it}}{23}}.$
    \item $ \Pr[ V_{it} < 0.65 \sigma_i^2] \le 3 e^{-0.01 \sigma_i^2 s_{it}}.$
\end{enumerate}
\end{theorem}

We present a proof sketch of this theorem in Appendix~\ref{app:tail}.
 
\subsection{Proof of Theorem~\ref{thm1}}
Let $t_0 = 4n$; we will ignore the first $t_0$ steps in the analysis, and they contribute $O(n)$ to the regret since the rewards are bounded in $[0,1]$. This implies $s_{it} \ge \frac{t}{2n}$ for $t \ge t_0$.

Let $q = \mbox{argmax}_i \mu_i$. Let $\mu^* = \mu_{q}$ and $\sigma^2_* = \sigma^2_{q}$. At any time step $t$, suppose a pair of arms $(i,j)$ with $\mu_i \ge \mu_j$ is played by the exploitation probes. Let $\Delta_j = \mu^* - \mu_j$. Suppose $\E[\max(X_i,X_j)] = \mu_{ij} \ge \mu^*$, then the expected regret is zero, otherwise, the expected regret is at most $\Delta_j$. We charge this regret to arm $j$.  We will now compute the probability with which a regret of $\Delta_j$ is charged to arm $j$ in time step $t$. Below, we will omit $t$ from the subscript when the connotation is obvious.

Let $\Delta_j = \ell \sigma^2_j = \rho \sigma^2_*$. We split the analysis into three cases. Note that the key hurdle with obtaining constant regret is that the tail bounds  in Theorem~\ref{thm:tail} that decay exponentially with time at rate depending on variance only hold when the deviation from the mean is at least a constant times the variance of the arm. We therefore split our analysis based on whether $\Delta_j$ is at least a constant times $\sigma^2_j$ or not. In the latter case, we use Lemma~\ref{lem:max} to argue that the regret is already non-positive. In the former case, we obtain exponentially decaying regret in time, which is sufficient to obtain overall constant regret. 

In the following, for notational brevity we may drop $t$ from $m_{jt}$, $V_{jt}$ and $s_{jt}$.

\medskip
\noindent {\bf Case 1. $\ell \le \frac{1}{2}$.} In this case, setting $X = X_j$ and $Y = X_i$ in Lemma~\ref{lem:max}, we have  
$$\mu_{ij} \ge \mu_j + \sigma^2_j/2 \ge \mu_j + \Delta_j \ge \mu^*.$$
Therefore, playing $(i,j)$ incurs non-positive regret.

\medskip
\noindent {\bf Case 2. $\ell \ge \frac{1}{2}$ and $\rho \ge \frac{1}{2}$.} Let $w = \mu_j + \frac{\Delta_j}{2}$. Recall the definition of $\UCB_{it}$ from Eq~(\ref{eq:UCB1}), and define the {\em good event} as $\UCB_j \le w$ {\bf and} $\UCB_{q} \ge w$. In this case, if $j$ is probed, then $q = i$, so that there is no regret. We upper bound the probability of the good event not happening by a union of {\em bad} events.

The first bad event is that $V_j > 0.65 \sigma^2_j$, where $V_j$ is as defined in Eq~(\ref{eq:sample2}). By Theorem~\ref{thm:tail}, there exists constant $c_1 > 0$ such that
$$ \Pr[V_j > 0.65 \sigma^2_j] \le  3 e^{- c_1 \ell \sigma^2_j s_j} = 3 e^{- c_1 \Delta_j s_j}.$$
Assuming this event do not happen, the second bad event is $\UCB_j > w$, which is equivalent to $m_j > \mu_j + \frac{\ell}{2} \sigma^2_j - \epsilon V_j,$ where $m_j$ is as defined in Eq~(\ref{eq:sample2}). Since $V_j \le 0.65 \sigma^2_j$, it suffices to bound the probability of the event
$$ m_j - \mu_j > (\ell / 2 - 0.65 \epsilon)  \sigma^2_j > 0.3 \ell \sigma^2_j,$$
where we have used $\epsilon = 0.1$. By Theorem~\ref{thm:tail}, there is a constant $c_2 > 0$ so that
$$ \Pr[\UCB_j > w | V_j \le \ell \sigma^2_j] \le 3 e^{- c_2 \ell \sigma_j^2 s_j} = 3 e^{-c_2 \Delta_j s_j}.$$
Analogously, the third bad event is that $\UCB_{q} < w = \mu_j + \frac{\Delta_j}{2} = \mu^* - \frac{\Delta_j}{2}$. This implies $m_{q} < \mu^* - \frac{\rho}{2} \sigma^2_*$. Repeating the same argument as the second bad event, we obtain
$$ \Pr[\UCB_{q} < w] \le 3 e^{- c_2 \rho \sigma^2_* s_{q}} = 3 e^{-c_2 \Delta_j s_{q}}.$$

Denote the regret charged to arm $j$ at step $t$ as $R_{jt}$. Let $c_3 = \frac{\min(c_1, c_2)}{2}$. Note next that $s_{jt}, s_{qt} \ge \frac{t}{2n}$. We now take the union bound over the three bad events above, and note that the regret charged to arm $j$ conditioned on the bad event is at most $\Delta_j$. This implies:
$$ \E[R_{jt}] \le 9 \Delta_j   e^{- c_3 \Delta_j \frac{t}{n}}.$$

\medskip
\noindent {\bf Case 3. $\ell \ge \frac{1}{2}$ and $\rho < \frac{1}{2}$.} We define the {\em good} event has having $\UCB_j \le \mu^*$, and $\UCB_{q} \ge \mu^*$. In this case, if arm $j$ is played, then arm $q$ is also played so that the regret is zero. As before, we upper bound the probability of the good event not happening by a union of {\em bad} events.

As in Case (2), $\UCB_j > \mu^*$ is captured by two bad events $V_j > \ell \sigma^2_j$ and $\UCB_j > \mu^*$ given $V_j \le \ell \sigma^2_j$. It is easy to check that the probability of these events are upper bounded by those derived for Case (2), so that: 
$$ \Pr[\UCB_j > \mu^*] \le 6 e^{- c_3 \ell \sigma_j^2 s_j} = 6 e^{- c_3\Delta_j s_j}.$$

To capture $\UCB_{q} < \mu^*$, we consider two other bad events $V_{q} < 0.65 \sigma^2_*$ and $\UCB_{q} < \mu^*$ given $V_{q} \geq 0.65 \sigma^2_*$.
Using Theorem~\ref{thm:tail} and the fact that $\Delta_j < \frac{\sigma^2_*}{2}$, we have:
$$ \Pr[ V_{q} < 0.65 \sigma^2_*] \le 3 e^{-0.01 \sigma^2_* s_{q}} \le 3 e^{-c_4 \Delta_j s_{q}}$$
for some constant $c_4 > 0$. Assume therefore that $V_{q} \ge 0.65 \sigma^2_*$. The last bad event is $\UCB_{q} < \mu^*$, which is equivalent to $\mu^* - m_{q} > \epsilon V_{q}$, which implies  $\mu^* - m_{q} > 0.065 \sigma^2_*$, since $\epsilon = 0.1$. Using Theorem~\ref{thm:tail} and the fact that $\Delta_j < \frac{\sigma^2_*}{2}$, there is a constant $c_5 > 0$ such that
$$ \Pr[\UCB_{q} < \mu^* | V_{q} \ge 0.65 \sigma^2_*] \le 3 e^{- \frac{0.065}{23} \sigma^2_* s_{q}} \le 3 e^{-c_5 \Delta_j s_{q}}.$$

As before, we take the union of all these bad events and set $c_6 = \frac{1}{2} \min(c_3, c_4, c_5)$ to obtain:
$$ \E[R_{jt}] \le 12 \Delta_j   e^{- c_6 \Delta_j \frac{t}{n}}.$$

Given the tail bounds derived in each of the three cases, by linearity of expectation over all time steps $t$ and sub-optimal arms $j$ to which the regret can be charged, we have:
$$ \E[\mbox{Regret}] \le  O(n) + \sum_{j \neq q} \sum_{t = t_0}^T 12 \Delta_j   e^{- c_6 \Delta_j \frac{t}{n}} = O(n^2),$$
where use a regret of $t_0 = O(n)$ for the first $t_0$ steps, and use linearity of expectation beyond that. This completes the proof of Theorem~\ref{thm1}.

\section{Handling Correlation Between Arms in the \all{} Model}
\label{sec:mab-dependent}
We now consider the \all{} model when the rewards on the arms can be correlated. In other words, the reward vector $\vec{r_t}$ at any time $t$ is drawn from a joint distribution over $[0,1]^n$. These draws are $i.i.d.$ across time steps.  Note that the analysis for the independent case presented above crucially needs Lemma~\ref{lem:max}, which does not hold when arms can be correlated. We now show a different algorithm and analysis (in the \all{} model) that uses $k=4$ probes and achieves a regret bound of $\widetilde{O}(T^{1/3})\cdot  \text{poly}(n)$. Again note that such a dependence on $T$ cannot be achieved in the standard bandit model with $k=1$ probes. 

\subsection{Correlation-Exploitation Algorithm}
As before the idea is to explore and exploit simultaneously, but the algorithm now plays pairs of arms, and thus can keep track of the {\em gain} offered by playing two arms simultaneously. More formally, the algorithm uses $k=4$ probes in every step, and consists of two exploit and two explore probes. 

\medskip
\noindent {\bf Explore Probes.} The explore probes pull every pair $(i,j)$ of arms in a round-robin fashion. Using these probes, the algorithm maintains estimates of the mean reward $\mu_i$, and additionally, estimates of quantities 
\begin{equation}
    G_{ji} := \mathbb{E}[ (X_j - X_i)_+]. \label{eq:def-gain}
\end{equation} 
Let us call the estimates $\hat{\mu}_i$ and $\hat{G}_{ji}$ respectively. The latter estimates the gain that arm $j$ offers over arm $i$ when played together. These estimates turn out to be crucial in handling correlations in the rewards. Note that since the $\hat{\mu}_i$ uses the same samples used to estimate $\hat{G}_{ji}$, these estimates can be dependent. 

\medskip
\noindent {\bf Exploit Probes.} For every arm $i$, define its ``partner'' as $\text{argmax}_j \hat{G}_{ji}$. The partner can change with time, and is a random variable that depends on the rewards obtained so far.  At every time $t$, for the exploit probes, the algorithm pulls the arm $i$ that has the highest value of $\hat{\mu}_i$ -- we call this the \emph{primary} arm at time $t$ -- and its partner. 

\subsection{Analysis}
We next turn to the analysis of the algorithm described above. Recall that $\Delta_i = \mu_q - \mu_i$, where $q$ is the arm with highest expected reward. We will show the following theorem. Though the improvement is not as impressive as for the independent reward case, we note that such a dependence on $T$ cannot be obtained with $k=1$ probes.

\begin{theorem}\label{thm:main3}
The regret of the correlation-exploitation algorithm is bounded by 
$$ \mbox{Regret} \le O \left(n^3\cdot \sum_{i \neq q} \sqrt{\frac{\log (1/\Delta_i)}{\Delta_i}} \right),$$
which implies a parameter independent regret bound of $\tilde{O}(n^{8/3} T^{1/3})$.
\end{theorem}

The rest of this section is devoted to proving the above theorem.  Towards the end, we will show that the analysis is tight and the dependence on $1/\sqrt{\Delta}$ cannot be improved for this algorithm.

Define $Z_{i,t}$ to be the random variable indicating if arm $i$ is the primary arm at time $t$.  The first observation is that for all arms with $\Delta_i >0$,
\begin{equation}\label{eq:zit}
 \Pr[Z_{i,t} =1] \le 2\exp(-t \Delta_i^2/4n).  
\end{equation}
To see this, note that in order to choose arm $i$ over arm $q$, we must have either $\hat{\mu_i} \ge \mu_i + \frac{\Delta_i}{2}$ or $\hat{\mu}_q \le \mu_q - \frac{\Delta_i}{2}$. Since we have at least $t/n$  {\em i.i.d.} samples for each arm and since the rewards are in $[0,1]$, the probability of each of these events can be bounded using Bernstein's inequality, and taking the union over the two events implies~\eqref{eq:zit}.

Next, we turn to the analysis of the quantities $G_{ij}$ and their estimates $\hat{G}_{ij}$.  The first observation is the following:

\begin{obs}\label{obs:gij} For any two arms $(i,j)$, we have
 $\mathbb{E}[\max(X_i, X_j)] = \mu_i + \mathbb{E}[(X_j - X_i)_+] = \mu_i + G_{ji}$. Furthermore, for every arm $i$, there exists $j$ such that $G_{ji} \ge \Delta_i$.
\end{obs}

This implies that when we play $i$, if we are able to identify its ``optimal partner'' $j$, then we will incur zero regret in expectation. However, since we only estimate $G_{ji}$, the actual regret can be higher. To analyze this difference, let us define $D_i (t)$ to be $\mathbb{E}[(X_j - X_i)_+]$, where $j$ is the partner of $i$ at time $t$ and the expectation is over the reward distribution. As $j$ is a random variable (depending on the rewards observed at times $t' < t$), so is $D_i(t)$. The key observation is that the expected regret at time $t$ conditioned on $i$ being the explore arm is bounded by $\Delta_i - D_i(t)$. 

Our overall approach is to bound the expected regret as
\[ \sum_i \sum_t \mathbb{E}[ Z_{i,t} (\Delta_i - D_i(t)) ].  \]
To bound this, we fix an index $i$ and analyze the sum over $t$. We can further bound the sum as:
\begin{equation}
    \sum_{t \le T_i} \mathbb{E}[\Delta_i - D_i(t)] + \sum_{t > T_i} \mathbb{E}[Z_{i,t}] \Delta_i,\label{eq:regret-toshow}
\end{equation}
where $T_i = \frac{4n^2 \log (1/\Delta_i)}{\Delta_i^2}$. The second summation is bounded easily using \eqref{eq:zit}:
\[ \sum_{t > T_i} \mathbb{E}[Z_{i,t}] \Delta_i \le 2 \Delta_i \int_{t = T_i}^{\infty} e^{-t\Delta_i^2 / 4n} ~dt = 2\Delta_i \cdot \frac{4n}{\Delta_i^2} \cdot e^{-T_i \Delta_i^2/4n} \le O(1). \]

Let us thus focus on the first sum. For the first $\frac{12 n^2}{\Delta_i}$ time steps, we simply bound the expectation by $\Delta_i$, which makes the sum add up to $O(n^2)$. As $t$ increases, the following lemma shows that $(\Delta_i - D_i(t))$ becomes much smaller than $\Delta_i$ with high probability.

\begin{lemma}\label{lem:regret-term}
Assume w.l.o.g. that $\Delta_i < 1$. Suppose that $t \ge 12 n^2 / \Delta_i$. Then for any $c \ge 0$,
\[ \Pr[ \Delta_i - D_i(t) \ge c] \le \Pr[ \Delta_i - D_i(t) > c\Delta_i ] \le n e^{-c^2 \Delta_i t/36 n^2}.\]
\end{lemma}
\begin{proof} 
First note that we can assume $0 < c < 1$; the bound for $c \ge 1$ is trivial. The proof proceeds in two parts, both of which use the fact that after $t$ steps, the explore arms (which perform round robin) use at least $(t/n^2)$ samples for computing each of the $\hat{G}_{ji}$. First, we show that for the optimal arm $q$,  since $G_{qi}\ge \Delta_i$ from Observation~\ref{obs:gij},
\[ \Pr[ \Delta_i - \hat{G}_{qi} > \frac{c}{2}\Delta_i] \le \Pr[ \hat{G}_{qi} < (1-\frac{c}{2}) G_{qi} ] \le e^{-c^2 \Delta_i t/12 n^2}.\]
This is a direct application of the standard Chernoff bound (e.g., part (1) of Theorem~\ref{app:chernoff}). Next, we consider any arm $j$ for which $G_{ji}$  is $\le \Delta_i (1-c)$. In this case, we wish to argue that $\hat{G}_{ij} < \Delta_i - \frac{c}{2} \Delta_i$ with high probability. For this, we consider two cases.

\noindent \emph{Case 1.} $G_{ji} < \Delta_i/4$. In this case, $\Delta_i - \frac{c}{2}\Delta_i > \Delta_i/2 > 2G_{ji}$, and thus to bound $\Pr[ \hat{G}_{ji} \ge 2G_{ji}]$, we can use the ``high deviation'' regime of Chernoff bounds (part (2) of Theorem~\ref{app:chernoff}) to conclude that 
\[ \Pr[ \hat{G}_{ji} \ge \Delta_i - \frac{c}{2} \Delta_i] \le \Pr[ \hat{G}_{ji} \ge \Delta_i/2 ] \le e^{-t\Delta_i/6n^2}. \]

\noindent \emph{Case 2.} $G_{ji} \ge \Delta_i/4$. In this case, since $\Delta_i - \frac{c}{2}\Delta_i \ge (1+\frac{c}{2})G_{ji}$, we can use a Chernoff bound again, to obtain
\[ \Pr[ \hat{G}_{ji} \ge \Delta_i - \frac{c}{2} \Delta_i] \le e^{-c^2 G_{ji} t/12 n^2} \le e^{-c^2 \Delta_i t/ 48 n^2}. \]

Combining the two parts and taking a union bound, we have that with probability at least $1 - n e^{-c^2 \Delta_i t/48 n^2}$, we have that (a) $\max_{j} \hat{G}_{ji} \ge \Delta_i - \frac{c}{2} \Delta_i$, and (b) the max is not attained by any $j$ with $G_{ij} \le G_{ji} (1-c)$. If both (a) and (b) hold, then $\Delta_i - D_i(t) \le c\Delta_i$, and this completes the proof of the lemma.
\end{proof}

Lemma~\ref{lem:regret-term} can be used to bound the first term of~\eqref{eq:regret-toshow}, using the following technical lemma.

\begin{lemma}\label{lem:tech1}
Let $n, \beta \ge 1$ be parameters, and $Y$ be a random variable that satisfies the condition:
\[ \forall c >0, \Pr[Y \ge c] \le n e^{-c^2 \beta}.\] 
Then $\mathbb{E}[Y] \le \frac{2n}{\sqrt{\beta}} $.
\end{lemma}
\begin{proof}
Since we are only interested in an upper bound on $\mathbb{E}[Y]$, we can ignore potential negative values of $Y$, and write

$$ \mathbb{E}[Y] \le \int_{c=0}^{\infty} \Pr[Y \ge c] ~dc \le n \int_{c=0}^{\infty} e^{-c^2 \beta} dc.$$
We then split the integrals into a sum over the intervals $c \in [0, \frac{1}{\sqrt{\beta}}],~[\frac{1}{\sqrt{\beta}}, \frac{2}{\sqrt{\beta}}], ~[\frac{2}{\sqrt{\beta}}, \frac{3}{\sqrt{\beta}}], \dots$. As the integrand in the $(i+1)$th interval is bounded by $\exp(-i^2)$, the sum can be bounded as desired.
\end{proof}

Using the lemmas, we can bound the first summation in~\eqref{eq:regret-toshow}. The main observation is that by using Lemma~\ref{lem:regret-term}, if $t = \beta \frac{48 n^2}{\Delta_i}$, the hypothesis of Lemma~\ref{lem:tech1} is satisfied for $Y = \Delta_i - D_i(t)$.  This implies that we can bound 
\[ \mathbb{E}[\Delta_i - D_i(t)] \le \frac{2n}{\sqrt{\beta}} = 12 n^2 \sqrt{\frac{\Delta_i}{t}}.  \]

Summing this between $t = \frac{48 n^2}{\Delta_i}$ (or even $t=1$) and $t = T_i = \frac{4n^2 \log (1/\Delta_i)}{\Delta_i^2}$, we obtain a bound of $O\left(  n^2 \sqrt{\Delta_i T_i} \right)$. Plugging in the value of $T_i$ then completes the proof of Theorem~\ref{thm:main3}.

\subsubsection{Tight Instance} 
We now show an instance where the above algorithm has regret $\Omega(n/\sqrt{\Delta})$.  There are three arms and $n-3$ dummy arms for large $n$. Arm $1$ has reward $X_t$ that is drawn {\em i.i.d.} from a Bernoulli distribution that is $1/3$ with probability $1/2$ and $2/3$ with probability $1/2$. Arm $2$ has reward $Y_t = X_t + A_t$, where $A_t$ is {\em i.i.d.} drawn from Bernoulli$(1/3, 3 \Delta)$. Arm $3$ has reward $Z_t  = X_t + B_t$, where $B_t$ is {\em i.i.d.} drawn from Bernoulli$(1/3, 3 \Delta(1-\sqrt{\Delta}))$, with $\Pr[B_t = 0 | A_t = 0] = 1$.

The dummy arms have reward zero at all time steps.  We assume that at every step, one pair of arms $(i,j)$ is sampled, and these samples are used to estimate $G_{ij}, G_{ji}, \mu_i$, and $\mu_j$. The dummy arms ensure the estimates $\hat{\mu_i}$ are approximately independent for all $i,j \in \{1,2,3\}$. Further, $\hat{G_{ji}}$ and $\hat{G_{ki}}$ are independent for all $i,j,k \in \{1,2,3\}$, since these estimates are constructed at different time steps. 

On this instance, the regret is with respect to arm $2$, with $\E[Y] = \frac{1}{2} + \Delta$. However, the construction of reward distributions ensures that $\E[\max(X,Y)] = \E[\max(Z,Y)] = \E[Y]$. Therefore, at any step, the expected regret of any strategy that plays pairs of arms is non-negative.

Using the tightness of Chernoff bounds on Bernoulli distributions, we can check that on this instance, with constant probability, the following two events happen for all $t \le \Omega(n/\Delta^2)$:
\begin{itemize}
    \item $\hat{\mu_1} > \max \left(\hat{\mu_2}, \hat{\mu_3} \right)$, so that the algorithm plays arm $1$;
    \item $\hat{G_{31}} > \hat{G_{21}}$, so that arm $3$ is the partner of arm $1$ and gets played.
\end{itemize}

In this event, the algorithm incurs regret $\Delta^{3/2}$ against the optimal arm $2$ each step, for a total regret of $\Omega(n/\sqrt{\Delta})$. This shows the analysis above is tight.
  
\section{Conclusion}
We conclude with some open questions. The main open question  is whether the stochastic assumptions are needed for the MAB results. In other words, can we obtain improved regret guarantees for the adversarial MAB problem~\cite{AuerCFS}. We make  progress in this direction with our results for correlated MAB in Section~\ref{sec:mab-dependent}; however, we believe our results even for this case can be improved.  

For the stochastic MAB problem, one intriguing open question is whether constant regret is possible for $k=2$ probes in the \all model. Note that Theorem~\ref{thm2} implies a regret of $O(n^2 \log T)$. 
However, unlike the celebrated Lai-Robbins result~\cite{LaiR} that shows the $\log T$ factor is necessary when $k=1$, we have not been able to show a lower bound requiring such dependence on $T$ for $k=2$, either for the \all model or for the \mmax model. We leave this as an interesting open question. Another interesting question is to extend our bandit results to the case with imperfect hints. 

At a higher level, it would be interesting to explore the power of a few probes in more complex bandit settings. One example is the linear contextual bandit problem~\cite{Abbasi,Li} where the stochastic arms correspond to latent variables. 
At any step, a decision space is given and the policy needs to choose a linear combination of these variables from the decision space, obtaining that linear combination of the reward of the arms as its reward. Now suppose the latent space of variables has small dimension, then does having multiple probes help with the regret bounds?
  
\bibliographystyle{abbrv}
\bibliography{refs}

\begin{thebibliography}{10}

\bibitem{Abbasi}
Y.~Abbasi-Yadkori, D.~P\'{a}l, and C.~Szepesv\'{a}ri.
\newblock Improved algorithms for linear stochastic bandits.
\newblock In {\em Proceedings of the 24th International Conference on Neural
  Information Processing Systems}, page 2312–2320, Red Hook, NY, USA, 2011.
  Curran Associates Inc.

\bibitem{Abernathy1}
J.~D. Abernethy, Y.~H. Jung, C.~Lee, A.~McMillan, and A.~Tewari.
\newblock Online learning via the differential privacy lens.
\newblock In {\em Advances in Neural Information Processing Systems},
  volume~32. Curran Associates, Inc., 2019.

\bibitem{agrawalG}
S.~Agrawal and N.~Goyal.
\newblock Analysis of thompson sampling for the multi-armed bandit problem.
\newblock In S.~Mannor, N.~Srebro, and R.~C. Williamson, editors, {\em
  Proceedings of the 25th Annual Conference on Learning Theory}, volume~23 of
  {\em Proceedings of Machine Learning Research}, pages 39.1--39.26, Edinburgh,
  Scotland, 25--27 Jun 2012. PMLR.

\bibitem{Audibert}
J.-Y. Audibert, R.~Munos, and C.~Szepesvári.
\newblock Exploration–exploitation tradeoff using variance estimates in
  multi-armed bandits.
\newblock {\em Theoretical Computer Science}, 410(19):1876--1902, 2009.

\bibitem{AuerCF}
P.~Auer, N.~Cesa-Bianchi, and P.~Fischer.
\newblock Finite-time analysis of the multiarmed bandit problem.
\newblock {\em Mach. Learn.}, 47(2–3):235–256, may 2002.

\bibitem{AuerCFS}
P.~Auer, N.~Cesa-Bianchi, Y.~Freund, and R.~E. Schapire.
\newblock The nonstochastic multiarmed bandit problem.
\newblock {\em SIAM J. Comput.}, 32(1):48–77, jan 2003.

\bibitem{AzarBKU}
Y.~Azar, A.~Z. Broder, A.~R. Karlin, and E.~Upfal.
\newblock Balanced allocations.
\newblock {\em SIAM J. Comput.}, 29(1):180–200, sep 1999.

\bibitem{BeyhaghiK19}
H.~Beyhaghi and R.~Kleinberg.
\newblock Pandora's problem with nonobligatory inspection.
\newblock In A.~Karlin, N.~Immorlica, and R.~Johari, editors, {\em Proceedings
  of the 2019 {ACM} Conference on Economics and Computation, {EC} 2019,
  Phoenix, AZ, USA, June 24-28, 2019}, pages 131--132. {ACM}, 2019.

\bibitem{bhaskara20a}
A.~Bhaskara, A.~Cutkosky, R.~Kumar, and M.~Purohit.
\newblock Online learning with imperfect hints.
\newblock In H.~D. III and A.~Singh, editors, {\em Proceedings of the 37th
  International Conference on Machine Learning}, volume 119 of {\em Proceedings
  of Machine Learning Research}, pages 822--831. PMLR, 13--18 Jul 2020.

\bibitem{bhaskara21}
A.~Bhaskara, A.~Cutkosky, R.~Kumar, and M.~Purohit.
\newblock Logarithmic regret from sublinear hints.
\newblock In M.~Ranzato, A.~Beygelzimer, Y.~Dauphin, P.~Liang, and J.~W.
  Vaughan, editors, {\em Advances in Neural Information Processing Systems},
  volume~34, pages 28222--28232. Curran Associates, Inc., 2021.

\bibitem{BhaskaraGKM20}
A.~Bhaskara, S.~Gollapudi, K.~Kollias, and K.~Munagala.
\newblock Adaptive probing policies for shortest path routing.
\newblock In H.~Larochelle, M.~Ranzato, R.~Hadsell, M.~Balcan, and H.~Lin,
  editors, {\em Advances in Neural Information Processing Systems 33: Annual
  Conference on Neural Information Processing Systems 2020, NeurIPS 2020,
  December 6-12, 2020, virtual}, 2020.

\bibitem{Blackwell}
D.~Blackwell.
\newblock {An analog of the minimax theorem for vector payoffs.}
\newblock {\em Pacific Journal of Mathematics}, 6(1):1 -- 8, 1956.

\bibitem{boucheron}
S.~Boucheron, G.~Lugosi, and P.~Massart.
\newblock {\em Concentration Inequalities - {A} Nonasymptotic Theory of
  Independence}.
\newblock Oxford University Press, 2013.

\bibitem{BubeckBook}
S.~Bubeck and N.~Cesa-Bianchi.
\newblock Regret analysis of stochastic and nonstochastic multi-armed bandit
  problems.
\newblock {\em Foundations and Trends® in Machine Learning}, 5(1):1--122,
  2012.

\bibitem{howto}
N.~Cesa-Bianchi, Y.~Freund, D.~Haussler, D.~P. Helmbold, R.~E. Schapire, and
  M.~K. Warmuth.
\newblock How to use expert advice.
\newblock {\em J. ACM}, 44(3):427–485, may 1997.

\bibitem{Dekel}
O.~Dekel, A.~Flajolet, N.~Haghtalab, and P.~Jaillet.
\newblock Online learning with a hint.
\newblock In I.~Guyon, U.~V. Luxburg, S.~Bengio, H.~Wallach, R.~Fergus,
  S.~Vishwanathan, and R.~Garnett, editors, {\em Advances in Neural Information
  Processing Systems}, volume~30. Curran Associates, Inc., 2017.

\bibitem{DeshpandeHK16}
A.~Deshpande, L.~Hellerstein, and D.~Kletenik.
\newblock Approximation algorithms for stochastic submodular set cover with
  applications to boolean function evaluation and min-knapsack.
\newblock {\em {ACM} Trans. Algorithms}, 12(3):42:1--42:28, 2016.

\bibitem{RothBook}
C.~Dwork and A.~Roth.
\newblock The algorithmic foundations of differential privacy.
\newblock {\em Found. Trends Theor. Comput. Sci.}, 9(3–4):211–407, 2014.

\bibitem{Freedman}
D.~A. Freedman.
\newblock On tail probabilities for martingales.
\newblock {\em The Annals of Probability}, 3(1):100--118, 1975.

\bibitem{GoelGM06}
A.~Goel, S.~Guha, and K.~Munagala.
\newblock Asking the right questions: {M}odel-driven optimization using probes.
\newblock In {\em Proc.\ of the 2006 {ACM} Symp.\ on Principles of Database
  Systems}, 2006.

\bibitem{GollapudiP19}
S.~Gollapudi and D.~Panigrahi.
\newblock Online algorithms for rent-or-buy with expert advice.
\newblock In {\em Proceedings of the 36th International Conference on Machine
  Learning, {ICML} 2019, 9-15 June 2019, Long Beach, California, {USA}}, pages
  2319--2327, 2019.

\bibitem{GolovinK}
D.~Golovin and A.~Krause.
\newblock Adaptive submodularity: Theory and applications in active learning
  and stochastic optimization.
\newblock {\em J. Artif. Int. Res.}, 42(1):427–486, Sept. 2011.

\bibitem{GuhaMS06}
S.~Guha, K.~Munagala, and S.~Sarkar.
\newblock Optimizing transmission rate in wireless channels using adaptive
  probes.
\newblock In {\em SIGMETRICS/Performance}, pages 381--382, 2006.

\bibitem{Hannan}
J.~Hannan.
\newblock {\em Approximation to {B}AYES risk in repeated play}, pages 97--140.
\newblock Princeton University Press, 2016.

\bibitem{hazan}
E.~Hazan.
\newblock Introduction to online convex optimization.
\newblock {\em CoRR}, abs/1909.05207, 2019.

\bibitem{im22a}
S.~Im, R.~Kumar, A.~Petety, and M.~Purohit.
\newblock Parsimonious learning-augmented caching.
\newblock In K.~Chaudhuri, S.~Jegelka, L.~Song, C.~Szepesvari, G.~Niu, and
  S.~Sabato, editors, {\em Proceedings of the 39th International Conference on
  Machine Learning}, volume 162 of {\em Proceedings of Machine Learning
  Research}, pages 9588--9601. PMLR, 17--23 Jul 2022.

\bibitem{KalaiV}
A.~Kalai and S.~Vempala.
\newblock Efficient algorithms for online decision problems.
\newblock {\em Journal of Computer and System Sciences}, 71(3):291--307, 2005.
\newblock Learning Theory 2003.

\bibitem{Kale}
S.~Kale.
\newblock Multiarmed bandits with limited expert advice.
\newblock {\em CoRR}, abs/1306.4653, 2013.

\bibitem{Kifer}
D.~Kifer, A.~Smith, and A.~Thakurta.
\newblock Private convex empirical risk minimization and high-dimensional
  regression.
\newblock In S.~Mannor, N.~Srebro, and R.~C. Williamson, editors, {\em
  Proceedings of the 25th Annual Conference on Learning Theory}, volume~23 of
  {\em Proceedings of Machine Learning Research}, pages 25.1--25.40, Edinburgh,
  Scotland, 25--27 Jun 2012. PMLR.

\bibitem{Krengel}
U.~Krengel and L.~Sucheston.
\newblock {Semiamarts and finite values}.
\newblock {\em Bulletin of the American Mathematical Society}, 83(4):745 --
  747, 1977.

\bibitem{Cascade1}
B.~Kveton, C.~Szepesvari, Z.~Wen, and A.~Ashkan.
\newblock Cascading bandits: Learning to rank in the cascade model.
\newblock In F.~Bach and D.~Blei, editors, {\em Proceedings of the 32nd
  International Conference on Machine Learning}, volume~37 of {\em Proceedings
  of Machine Learning Research}, pages 767--776, Lille, France, 07--09 Jul
  2015. PMLR.

\bibitem{LaiR}
T.~Lai and H.~Robbins.
\newblock Asymptotically efficient adaptive allocation rules.
\newblock {\em Advances in Applied Mathematics}, 6(1):4--22, 1985.

\bibitem{LattanziLMV20}
S.~Lattanzi, T.~Lavastida, B.~Moseley, and S.~Vassilvitskii.
\newblock Online scheduling via learned weights.
\newblock In S.~Chawla, editor, {\em Proceedings of the 2020 {ACM-SIAM}
  Symposium on Discrete Algorithms, {SODA} 2020, Salt Lake City, UT, USA,
  January 5-8, 2020}, pages 1859--1877. {SIAM}, 2020.

\bibitem{LattimoreBook}
T.~Lattimore and C.~Szepesvári.
\newblock {\em Bandit Algorithms}.
\newblock Cambridge University Press, 2020.

\bibitem{Li}
L.~Li, W.~Chu, J.~Langford, and R.~E. Schapire.
\newblock A contextual-bandit approach to personalized news article
  recommendation.
\newblock In {\em Proceedings of the 19th International Conference on World
  Wide Web}, page 661–670, New York, NY, USA, 2010. Association for Computing
  Machinery.

\bibitem{LittlestoneW}
N.~Littlestone and M.~Warmuth.
\newblock The weighted majority algorithm.
\newblock {\em Information and Computation}, 108(2):212--261, 1994.

\bibitem{Liu}
Z.~Liu, S.~Parthasarathy, A.~Ranganathan, and H.~Yang.
\newblock Near-optimal algorithms for shared filter evaluation in data stream
  systems.
\newblock In {\em Proceedings of the 2008 ACM SIGMOD International Conference
  on Management of Data}, page 133–146, New York, NY, USA, 2008.

\bibitem{lykouris2018competitive}
T.~Lykouris and S.~Vassilvtiskii.
\newblock Competitive caching with machine learned advice.
\newblock In {\em International Conference on Machine Learning}, pages
  3302--3311, 2018.

\bibitem{logit}
D.~McFadden.
\newblock {\em Conditional Logit Analysis of Qualitative Choice Behavior}.
\newblock BART impact studies final report series: Traveler behavior studies.
  Institute of Urban and Regional Development, University of California, 1973.

\bibitem{MitzenmacherV20}
M.~Mitzenmacher and S.~Vassilvitskii.
\newblock Algorithms with predictions.
\newblock In T.~Roughgarden, editor, {\em Beyond the Worst-Case Analysis of
  Algorithms}, pages 646--662. Cambridge University Press, 2020.

\bibitem{motwani-raghavan}
R.~Motwani and P.~Raghavan.
\newblock {\em Randomized Algorithms}.
\newblock Cambridge University Press, 1995.

\bibitem{garbage}
S.~Mukhopadhyay, S.~Sahoo, and A.~Sinha.
\newblock k-experts - online policies and fundamental limits.
\newblock {\em CoRR}, abs/2110.07881, 2021.

\bibitem{MunagalaBMW05}
K.~Munagala, S.~Babu, R.~Motwani, and J.~Widom.
\newblock The pipelined set cover problem.
\newblock {\em Proc. Intl. Conf. Database Theory}, 2005.

\bibitem{purohit2018improving}
M.~Purohit, Z.~Svitkina, and R.~Kumar.
\newblock Improving online algorithms via ml predictions.
\newblock In {\em Advances in Neural Information Processing Systems}, pages
  9661--9670, 2018.

\bibitem{RakhlinSridharan}
A.~Rakhlin and K.~Sridharan.
\newblock Online learning with predictable sequences.
\newblock In S.~Shalev{-}Shwartz and I.~Steinwart, editors, {\em {COLT} 2013 -
  The 26th Annual Conference on Learning Theory, June 12-14, 2013, Princeton
  University, NJ, {USA}}, volume~30 of {\em {JMLR} Workshop and Conference
  Proceedings}, pages 993--1019. JMLR.org, 2013.

\bibitem{Robbins}
H.~Robbins.
\newblock Some aspects of the sequential design of experiments.
\newblock {\em Bulletin of the American Mathematical Society}, 58(5):527 --
  535, 1952.

\bibitem{SamuelCahn}
E.~Samuel-Cahn.
\newblock Comparison of threshold stop rules and maximum for independent
  nonnegative random variables.
\newblock {\em The Annals of Probability}, 12(4):1213 -- 1216, 1984.

\bibitem{Seldin}
Y.~Seldin, K.~Crammer, and P.~Bartlett.
\newblock Open problem: Adversarial multiarmed bandits with limited advice.
\newblock In S.~Shalev-Shwartz and I.~Steinwart, editors, {\em Proceedings of
  the 26th Annual Conference on Learning Theory}, volume~30 of {\em Proceedings
  of Machine Learning Research}, pages 1067--1072, Princeton, NJ, USA, 12--14
  Jun 2013. PMLR.

\bibitem{SlivkinsBook}
A.~Slivkins.
\newblock Introduction to multi-armed bandits.
\newblock {\em Foundations and Trends® in Machine Learning}, 12(1-2):1--286,
  2019.

\bibitem{Slivkins}
A.~Slivkins, F.~Radlinski, and S.~Gollapudi.
\newblock Ranked bandits in metric spaces: Learning diverse rankings over large
  document collections.
\newblock {\em J. Mach. Learn. Res.}, 14(1):399–436, feb 2013.

\bibitem{SteinhardtLiang}
J.~Steinhardt and P.~Liang.
\newblock Adaptivity and optimism: An improved exponentiated gradient
  algorithm.
\newblock In {\em Proc. the 31th International Conference on Machine Learning,
  {ICML}}, volume~32 of {\em {JMLR} Workshop and Conference Proceedings}, pages
  1593--1601, 2014.

\bibitem{StreeterG}
M.~Streeter and D.~Golovin.
\newblock An online algorithm for maximizing submodular functions.
\newblock In D.~Koller, D.~Schuurmans, Y.~Bengio, and L.~Bottou, editors, {\em
  Advances in Neural Information Processing Systems}, volume~21. Curran
  Associates, Inc., 2008.

\bibitem{Thompson}
W.~R. Thompson.
\newblock On the likelihood that one unknown probability exceeds another in
  view of the evidence of two samples.
\newblock {\em Biometrika}, 25(3/4):285--294, 1933.

\bibitem{Wald}
A.~Wald.
\newblock {\em Sequential analysis}.
\newblock John Wiley, 1947.

\bibitem{WeiLuo}
C.~Wei and H.~Luo.
\newblock More adaptive algorithms for adversarial bandits.
\newblock In S.~Bubeck, V.~Perchet, and P.~Rigollet, editors, {\em Conference
  On Learning Theory, {COLT} 2018}, volume~75 of {\em Proceedings of Machine
  Learning Research}, pages 1263--1291. {PMLR}, 2018.

\bibitem{Pandora}
M.~L. Weitzman.
\newblock Optimal search for the best alternative.
\newblock {\em Econometrica}, 47(3):641--654, 1979.

\bibitem{Yue12}
Y.~Yue, J.~Broder, R.~Kleinberg, and T.~Joachims.
\newblock The k-armed dueling bandits problem.
\newblock {\em Journal of Computer and System Sciences}, 78(5):1538--1556,
  2012.

\bibitem{Zinkevich}
M.~Zinkevich.
\newblock Online convex programming and generalized infinitesimal gradient
  ascent.
\newblock In {\em Proceedings of the Twentieth International Conference on
  International Conference on Machine Learning}, ICML'03, page 928–935. AAAI
  Press, 2003.

\bibitem{CJW20}
J.~Zuo, X.~Zhang, and C.~Joe-Wong.
\newblock Observe before play: Multi-armed bandit with pre-observations.
\newblock {\em Proceedings of the AAAI Conference on Artificial Intelligence},
  34(04):7023--7030, Apr. 2020.

\end{thebibliography}

\appendix

\section{Extension to Experts and Online Convex Optimization}\label{app:extension1}
In this section, we extend the results in Section~\ref{sec:experts} to the experts and the online convex optimization settings. Note that though the experts problem is a special case of online linear optimization, we present an improved $O(\ln n)$ regret bound, where $n$ is the number of arms. We note that the extension of these results to the setting with $B$ imperfect hints (Section~\ref{sec:imperfect}) to yield regret that depends on $O(\sqrt{B+1})$ is straightforward and omitted.

\subsection{Experts Setting: Logarithmic Regret}\label{app:extension-experts}
In the special case of experts, the set $\W$ is the $d$-dimensional unit simplex, where the vertices (or dimensions) are called ``arms''. Playing an arm $i$ at time $t$ incurs loss $\ell^t_i$, and the algorithm should play one arm each time step, incurring its loss and subsequently learning the losses of all arms that step. The regret compares against the loss of choosing a single arm for all time steps, albeit with full information about the losses.  

In the probe model, the algorithm can probe two arms $a^t, b^t$ at step $t$ to find out $\mbox{argmin}_{w = a^t,b^t} \ell^t_{w}$. It subsequently plays an arm to incur its loss, finally learning the losses of all arms that step. Though one can directly use \lap for this setting, we present an improved regret bound.

\medskip \noindent {\bf Algorithm.} For this setting, instead of sampling from the {\tt Laplace} distribution, we sample each $x_j$ from a {\tt Gumbel} distribution with location $\mu = 0$ and scale $\beta = \frac{1}{\eta}$, where $\eta \in (0,0.4]$ is a constant. This distribution has CDF $F_Y(z) = \exp(-\exp(-\eta z))$, where the support of $z$ is all reals. 

For the {\tt Gumbel} distribution, using standard results~\cite{logit}, the \lap algorithm becomes a modification of the classical {\sc Hedge} algorithm. Recall that the {\sc Hedge} algorithm maintains a weight $W^t_i$ for each arm (dimension) $i$. At time $t$, the algorithm chooses arm $i$ to play with probability proportional to $W^t_i$. Subsequent to observing the losses at this step, it updates the weight of each arm $i$ as: 
$$ W^{t+1}_i = W^t_i \exp(- \eta \ell^t_i).$$

Our algorithm {\em Hedge with Choice} ({\sc HwC}) works as follows. At step $t$, let $\D$ denote the distribution over arms, where arm $i$ has probability $p^t_i = \frac{W^t_i}{\sum_{j=1}^n W^t_j}$. Note that the classical {\sc Hedge} algorithm plays arm $i$ with probability $p^t_i$. Instead, our algorithm independently samples $a^t, b^t$ from $\D$. Let $A^t = \ell^t_{a^t}$ and $B^t = \ell^t_{b^t}$. The algorithm probes to learn $w^t = \mbox{argmin}_{a^t, b^t} \{A^t, B^t \}$ and subsequently plays this arm, incurring loss $\min(A_t, B_t)$ at time $t$. 

\medskip \noindent {\bf Analysis.} The analysis is the same as before. We define \btrl using the {\tt Gumbel} distribution, and denote its loss at step $t$ as $C^t$. In Lemma~\ref{lem:kv}, we set $D = 2$, since the space $\W$ is the set of vertices of a $d$-dimensional unit simplex. Further, it can be easily shown~\cite{Abernathy1} that $\E[\max_j |x_j|] = O\left(\frac{\ln n}{\eta}\right)$. Therefore we have:

\begin{corollary}
\label{cor5}
For any $\eta > 0$, the regret of the \btrl algorithm is $O\left( \frac{\ln n}{\eta}\right)$.
\end{corollary}

Lemma~\ref{lem:dp} holds for this setting as well, this is shown in~\cite{Abernathy1} and we present a simple proof below. 
Formally, for the distributions $\D_1$ and $\D_2$ capturing the choice of arm $a^t$ (resp. $b^t$) and $c^t$, we have:

\begin{lemma}
For any arm $i \in \{1,2,\ldots, d\}$, we have:
$$ \exp(- \eta) \le \frac{\Pr[\D_1 = i]}{\Pr[\D_2 = i]} \le \exp(\eta) $$
\end{lemma}
\begin{proof} For any arm $i$, using the update rule of {\sc Hedge}, we have:
$$\frac{\Pr[\D_1 = i]}{\Pr[\D_2 = i]}  = \exp(\eta \ell^t_i) \frac{\sum_{j=1}^d \exp(-\eta L^{t-1}_j) \cdot \exp(- \eta \ell^t_j)}{\sum_{j=1}^d \exp(-\eta L^{t-1}_j)} \in [\exp(-\eta), \exp(\eta)],$$
where we have used $\ell^t_j \in [0,1]$ for all $j$.
\end{proof}

Combining this with Lemma~\ref{lem:min}, the regret of \hwc is at most that of the \btrl algorithm. Using Corollary~\ref{cor5} shows the following theorem:

\begin{theorem}
For experts setting, the \hwc algorithm with constant $\eta \in (0,0.4]$ has regret $O(\ln n)$.
\end{theorem}

\subsection{Online Convex Optimization}\label{app:extension-oco}
This setting is the same as online linear optimization, except that at time $t$, after choosing action $w^t$, the algorithm is given a convex loss function $\ell^t$ and it incurs loss $\ell^t(w^t)$. We assume $\W$ itself is a closed, convex, and compact subset of $[-1,1]^d$, and that the loss functions $\ell^t$ are generated by an oblivious adversary. We further assume that the norm of the gradient of $\ell^t$ is bounded by $\beta$, and its Hessian has eigenvalues upper bounded by $\gamma$.

As before, the algorithm can choose two points $a^t, b^t \in \W$ and obtain $\ell^t(a^t)$ and $\ell^t(b^t)$. It subsequently chooses an action $w^t \in \W$, incurring loss $\ell^t(w^t)$ and learning the loss function $\ell^t$.

\newcommand{\cwc}{{\sc CwC}\ }

\medskip \noindent {\bf Convex with Choice ({\sc CwC}) Algorithm.} We adapt the algorithm for differentially private ERM~\cite{Kifer} to online learning in a fashion similar to~\cite{Abernathy1}. Specifically, we modify the algorithm to use the better of two regularized outcomes. Formally, let $L^{t-1} = \sum_{q=1}^{t-1} \ell^q$ denote the sum of the loss functions till time $t$. For $\eta \le 0.4$ being a constant, the algorithm performs these steps at time $t$. 

\begin{itemize}
\item Choose $x \in \mathbb{R}^d$ from the {\tt Gamma} distribution with density $f(x) \propto \exp\left(-\eta \frac{\Vert x \Vert_2}{\beta}\right)$ and compute the regularized optimum
$$a^t = \mbox{argmin}_{w \in \W} \left(L^{t-1}(w) + \langle x, w \rangle + \frac{\gamma}{\eta} \Vert w \Vert_2^2\right).$$ 

\item Repeat the above step choosing $y$ independently from the same {\tt Gamma} distribution, and using $y$ instead of $x$ to compute the regularized optimum $b^t$. 
\item Let $A^t = \ell^t(a^t)$ and $B^t = \ell^t(b^t)$. Probe to learn $w^t = \mbox{argmin}_{a^t, b^t} \{ A^t, B^t \}$. 
\item Play $w^t$ at time $t$, incurring actual loss $\min(A^t, B^t)$. 
\end{itemize}

\medskip \noindent {\bf Analysis.} Our analysis is essentially the same as that for linear optimization. As before, define the \btrl algorithm that chooses $z$ independently from the same {\tt Gamma} distribution, and sets its action at time $t$ as:
$$ c^t = \mbox{argmin}_{w \in \W} \left(L^{t}(w) + \langle x, w \rangle + \frac{\gamma}{\eta} \Vert w \Vert_2^2\right).$$

The following lemma is implicit in the proof of Theorem 3.2 in~\cite{Abernathy1}, and simply applies the ``follow the leader" analysis of~\cite{KalaiV} to the gradient of the loss functions as done in~\cite{Zinkevich}.

\begin{lemma}[\cite{Abernathy1}]
\label{lem:btrl2}
For constant $\eta > 0$, the regret of \btrl is $O\left(\gamma d + \beta d^{3/2} \right)$.
\end{lemma}

As before, focus on time $t$ and let $\D_1$ denote the distribution over $a^t$ (resp. $b^t$), and let $\D_2$ denote the distribution over $c^t$. Theorem 2 of~\cite{Kifer} implies that $\D_1 $ and $\D_2$ satisfy $\eta$-differential privacy, that is, Lemma~\ref{lem:dp} holds. This implies that Lemma~\ref{lem:min} holds as well, so that the regret of \cwc is upper bounded by that of the \btrl algorithm. Combining with Lemma~\ref{lem:btrl2} shows the following theorem:

\begin{theorem}
\label{thm:main2}
For constant $\eta \in (0,0.4]$, the regret of the \cwc algorithm is $O\left(\gamma d + \beta d^{3/2} \right)$.
\end{theorem}

Note that for online linear optimization, if each $\ell^t \in [0,1]^d$, then $\beta \le \sqrt{d}$ and we recover the $\tilde{O}(d^2)$ regret bound in Theorem~\ref{thm:main1}.

\section{Concentration Bounds}
Our proofs use standard Chernoff bounds (see, e.g., \cite{boucheron, motwani-raghavan}). For completeness, we state the version we use.

\begin{theorem}\label{app:chernoff}
Let $X_1, X_2, \dots, X_n$ be independent random variables with support $[0,1]$. Let $p_i = \E[X_i]$, and let $\mu = \sum_i p_i$. Then we have the following:
\begin{enumerate}
    \item (Small deviation) For any $\delta \in [0,1]$,
    \begin{align*}
        \Pr[ X \ge (1+\delta) \mu] \le e^{-\mu \delta^2 / 3}, \\ 
        \Pr[ X \le (1-\delta) \mu] \le e^{-\mu \delta^2/3}.
    \end{align*}
    \item (Large deviation) For $\delta \ge 1$, we have
    \[ \Pr [ X \ge (1+\delta) \mu] \le e^{-\mu \delta/3}. \]
\end{enumerate}
\end{theorem}

\subsection{Proof Sketch of Theorem~\ref{thm:tail}}
\label{app:tail}
We present some details on the proof of Theorem~\ref{thm:tail}. 
In order to avoid confusion, we omit the subscript $t$ from $m_{it}, V_{it}, s_{it}$ and denote them  by $m_i, V_i, s_i$ respectively. Since we map the notation to that in~\cite{Audibert}, we will reserve the notations $n,t$ for the corresponding terms in that work. 

The following lemma bounding the sample mean follows from Eq~(48) in~\cite{Audibert} with $n=t$  and using the fact that $s_i \leq n$.

\begin{lemma}
\label{lem:bennett1}
For any $x \ge 0$, with probability $1 - 3 e^{-x}$, we have: $  \left| \mu_i - m_i \right| \le \sqrt{ \frac{2 \sigma_i^2 x}{s_i}} + \frac{x}{3s_i}.$
\end{lemma}

This yields the following corollary.

\begin{corollary}
\label{cor1}
For $q \ge \frac{1}{18}$, we have: $ \Pr\left[ |m_i - \mu_i| > q \sigma^2_i \right] \le 3 e^{- \frac{q \sigma^2_i s_i}{23}}.$
\end{corollary}
\begin{proof}
For $r \ge \frac{1}{300}$, set $x = r \sigma_i^2 s_i$ in Lemma~\ref{lem:bennett1}. Then, with probability  $1 - 3 e^{-r \sigma_i^2 s_i}$, we have $ | m_i - \mu_i| \le \sigma_i^2 (\sqrt{2r} + r/3) \le 23 \sigma^2_i r.$ Setting $q = 23r$ completes the proof.
\end{proof}

This shows the first part of Theorem~\ref{thm:tail}. A similar bound holds for the sample variance. 

\begin{lemma} 
For any $x \ge 0$, with probability $1 - 3 e^{-x}$, we have: $ V_i \le \sigma^2_i + \sqrt{ \frac{2 \sigma_i^2 x}{s_i}} + \frac{x}{3s_i}.$
\end{lemma}
\begin{proof}
In Eq~(43) of~\cite{Audibert}, set $U_{iq} = (X_{iq} - \mu_i)^2$. We have $\E[U_i^2] \le \E[U_i] = \sigma_i^2$ since $U_{iq} \in [0,1]$. Further, $b''_{+} = 1 - \E[U_i] \le 1$. Let $m_i = \frac{\sum_{q=1}^{s_i} X_{iq}}{s_i}$ and $V_i = \frac{\sum_{q=1}^{s_i} (X_{iq} - m_i)^2}{s_i}$. Setting $n = t = s_i$, with probability $1 - e^{-x}$, we have:
$$ \frac{\sum_{q=1}^{s_i} U_{iq}}{s_i} = \frac{\sum_{q=1}^{s_i} X_{iq}^2}{s_i} - 2 \mu_i m_i + \mu_i^2 \le \sigma_i^2 + \sqrt{\frac{2 \sigma_i^2 x}{s_i}} + \frac{x}{3 s_i}.$$
Note that the LHS is simply $V_i + (\mu_i - m_i)^2 \ge V_i$. The lemma now follows.
\end{proof}

This yields the following corollary using the same proof method as Corollary~\ref{cor1}.

\begin{corollary}
\label{cor3}
For $q \ge \frac{1}{18}$, we have: $ \Pr\left[ V_i > (1+q) \sigma^2_i \right] \le 3 e^{- \frac{q \sigma^2_i s_i}{23}}.$
\end{corollary}

This shows the second part of Theorem~\ref{thm:tail}. Finally, we  have the following inequality that bounds $V_i$ in the other direction. This inequality follows from Eq~(50) in~\cite{Audibert} by simply 
setting $x = 0.01 \sigma_i^2 s_i$, $n=t = s_i$, and $L = nx / t^2 = 0.01 \sigma_i^2$,
$$ \Pr[ V_i < 0.65 \sigma_i^2] \le 3 e^{-0.01 \sigma_i^2 s_i}.$$
This completes the proof of Theorem~\ref{thm:tail}.

\end{document}